\documentclass[journal]{IEEEtran}    

\usepackage{times}
\usepackage{amssymb}
\usepackage{amsmath}
\usepackage{amsthm}
\usepackage{amsfonts}
\usepackage{txfonts}
\usepackage{graphicx}
\usepackage{algorithm}
\usepackage{algorithmic}
\usepackage{subfig}

\newtheorem{theorem}{Theorem}
\newtheorem{lemma}[theorem]{Lemma}

\newtheorem{definition}{Definition}

\newcommand{\ignore}[1]{{}}
\newcommand{\argmax}{\operatornamewithlimits{argmax}}

\begin{document}

\title{A Hybrid Pricing Framework for TV \\White Space Database}

\author{Xiaojun~Feng,~\IEEEmembership{Student~Member,~IEEE,}
        Qian~Zhang,~\IEEEmembership{Fellow,~IEEE,}
        Jin~Zhang,~\IEEEmembership{Member,~IEEE,}
\thanks{X. Feng and Q. Zhang are with the Department of Computer Science and Engineering, Hong Kong University of Science and Technology, Hong Kong.  e-mail: \{xfeng, qianzh\}@cse.ust.hk.}
\thanks{J. Zhang is with the HKUST Fok Ying Tung Research Institute, Hong Kong. e-mail: jinzh@ust.hk}}

\maketitle
\begin{abstract}
According to the recent rulings of the Federal Communications Commission (FCC), TV white spaces (TVWS) can now be accessed by secondary users (SUs) after a list of vacant TV channels is obtained via a geo-location database. Proper business models are therefore essential for database operators to manage geo-location databases.
Database access can be simultaneously priced under two different schemes: the registration scheme and the service plan scheme.
In the registration scheme, the database reserves part of the TV bandwidth for registered White Space Devices (WSDs). In the service plan scheme, the WSDs are charged according to their queries.

In this paper, we investigate the business model for the TVWS database under a hybrid pricing scheme. We consider the scenario where a database operator employs both the registration scheme and the service plan scheme to serve the SUs. The SUs' choices of different pricing schemes are modeled as a non-cooperative game and we derive distributed algorithms to achieve Nash Equilibrium (NE). Considering the NE of the SUs, the database operator optimally determines pricing parameters for both pricing schemes in terms of bandwidth reservation, registration fee and query plans.
\end{abstract}


\section{Introduction}

Recently, the FCC has released the TVWS for secondary access \cite{FCCrule10}\cite{FCCrule12} under a database-assisted architecture, where there are several geo-location databases providing spectrum availability of the TV channels. These databases will be managed by database operators (DOs) approved by the FCC. To manage the operation costs, proper business models are essential for DOs.

The FCC allows DOs to determine their own pricing schemes \cite{FCCrule10} and there are two different ways in which SUs can assess the TVWS with the help of a database. SUs can register their WSDs in the database in a soft-licence style \cite{crwoncom12}. Part of the available TV spectrum is then reserved for the registered SUs \cite{ReserveFCC}\cite{ReserveOfcom}. Unregistered SUs can also access TVWS in a purely secondary manner. For instance, an SU can first connect to the database and upload WSD information such as location and transmission power and then obtain a list of available channels from the database.

As a result, two different pricing schemes can be employed, one for registered and the other for unregistered SUs, respectively. Registered SUs pay a registration fee to DOs and access the reserved bandwidth exclusively. This pricing scheme can be referred to as the \emph{registration scheme}. Unregistered SUs query the database only when they are in need of TV spectrum. DOs charge them according to the number of database queries they make. This pricing scheme is referred to as the \emph{service plan scheme}. The co-existence of multiple pricing channels allows DOs to better manage their costs and provides different service qualities to different types of SUs. For example, the registration scheme can be adopted by SUs providing rural broadband or smart metering services, since the reserved bandwidth may suffer from less severe interference. On the other hand, the service plan scheme suits the temporary utilization of TVWS such as home networking.

To harvest the advantages of both pricing schemes and maximize the profit of DOs, two challenges need to be addressed. First, how should DOs determine pricing parameters for each scheme? With limited available TV bandwidth, DOs need to decide how much to allocate to each pricing channel. Also, the registration fee in the registration scheme and the price for a certain amount of queries should be determined. Second, how should SUs choose between the two schemes? Both schemes have their pros and cons for different types of SUs. The two challenges are coupled together. The decisions of SUs on which pricing scheme to choose can affect the profit obtained by DOs while the pricing parameters designed by DOs ensure SUs have different preferences for either scheme.

In this paper, we focus on DO's hybrid pricing scheme design considering both the registration and the service plan scheme. To the best of our knowledge, there are no existing works on geo-location database considering a hybrid pricing model \cite{database_dyspan08}-\cite{database_icc12}.
We consider one DO and multiple types of SUs. The SUs can strategically choose between the two pricing schemes. Unlike many existing works consider no SU strategies when there exist multiple pricing schemes, in this paper, we assume users can have their own choices other than being directly classified into either pricing scheme. We argue that especially for a new service like database-based networking, users will consider seriously of the benefits from each scheme and other players' responses.

In this paper, we consider a two-stage pricing framework. At Stage I, the DO announces the amount of bandwidth to be reserved and the registration fee for the registered SUs and then the SUs choose whether to register or not. At Stage II, the DO announces a set of service plans for the unregistered SUs to choose from.

For the SUs, they decide which pricing scheme to choose given the announced pricing schemes. If the service plan scheme is adopted, SUs should further decide one particular plan to buy. In this paper, we consider two different SU scenarios. In the non-strategic case, SUs have fixed their pricing scheme choices. In the strategic case, SUs can compete with each other in choosing either pricing schemes.
In the later case, the competition among the SUs is modeled with the non-cooperative game theory since the choices of the other SUs can affect an SU's utility owing to the sharing of TVWS.

For the DO, the problem is to optimally allocate bandwidth and design pricing parameters based on estimated actions of the SUs. In this paper, we consider the DO has either complete or incomplete information of the SUs. By complete information, we mean that the DO knows the exact \emph{type} of each SU. The type of SU relates with its channel quality, valuation for the spectrum et. al. In the case of incomplete information, the DO knows only a distribution of the types. We model the optimal service plan design problem with contract theory. Different service plans are considered to be different contract items and an optimal contract is determined based on the knowledge of the SU types. To optimally choose pricing parameters, one challenge to tackle is how the DO estimate the possible actions of the SUs, especially for the strategic case. There may be multiple possible equilibriums existed in the game. We solve this challenge by exploring the nature of our problem and propose computationally feasible algorithm to estimate the outcome of the game.

The major contributions of this paper are summarized as follows:
\begin{itemize}
  \item We propose a hybrid pricing scheme for the DO considering the heterogeneity of SUs' types. As far as we know, it is the first work considering hybrid pricing schemes for TVWS database.
  \item We model the competition among the strategic SUs as a non-cooperative game. We prove the existence of the Nash Equilibrium (NE) under both the complete and incomplete information cases. By exploring the nature of the problem, we design computational feasible distributed algorithms for the SUs to achieve an NE with bounded iterations.
  \item We propose algorithms for the DO to optimally decide pricing parameters. We formulate the service plan design with contract theory and derive optimal contract items under the complete information and sub-optimal items for and incomplete information scenario.
\end{itemize}

The rest of this paper is organized as follows. In Section \ref{sec:model}, we describe the pricing framework and detailed system model. Problem formulation is presented in Section \ref{sec:formulation}. In Section \ref{sec:non-strategic_complete} we study the optimal pricing solution for non-strategic SUs under the complete information scenario as a baseline case. Then in Section \ref{sec:strategic_complete} and Section \ref{sec:strategic_incomplete}, we consider the SUs to be strategic players, under complete and incomplete information scenario, respectively. Numerical results are given in Section \ref{sec:simulation}. Related works are further reviewed in Section \ref{sec:related_works}. Finally, Section \ref{sec:conclusion} concludes the paper.

\section{System Model and Pricing Framework}
\label{sec:model}

In this section, after an overview of the system parameters, we introduce the big picture of the proposed pricing framework. Then we further detail the model for the DO and the SUs, respectively.
Key notations are summarized in Table \ref{tab:Notations}.

\subsection{System Parameters}
\label{subsec:parameter}

We consider the scenario with one DO and $N$ SUs denoted as $\mathcal{N}=\{u_1, u_2, \cdots, u_N\}$. We assume $N$ is public information available to the DO and the SUs. For the ease of analysis, we assume all the SUs are in the same contention domain \cite{contention_domain}. FCC requires the SUs to periodically access the database \cite{FCCrule10}. In this paper, we consider a time duration of $M$ periods. In each period, before accessing the TV channels, the SUs should connect to the database to obtain a list of available channels. We assume the total available TV channels in each period have an expected bandwidth of $B$ and the bandwidth of a channel is $b_0$. The number of available channels is then $B/b_0$.

\begin{table}[tp]
\renewcommand{\arraystretch}{1.0}
\caption{Key notations in this paper.}
\label{tab:Notations}
\centering
\begin{tabular}{|c|p{2.8in}|}
  \hline
  $N$ & The number of SUs\\
  \hline
  $M$ & The number of time periods considered\\
  \hline
  $u_i$ & The secondary user with index $i$\\
  \hline
  $a_i$ & The pricing scheme selected by $u_i$, 0: the registration scheme; 1: the service plan scheme\\
  \hline
  $\theta_i$ & The $type$ of $u_i$\\
  \hline
  $B$ & The expected bandwidth of available TVWS\\
  \hline
  $B_R$ & The expected bandwidth reserved for registered SUs\\
  \hline
  $r$ & The registration fee\\
  \hline
  $q_i$ & The number of database queries in a query plan $(q_i, p_i)$\\
  \hline
  $p_i$ & The wholesale price in the query plan $(q_i, p_i)$\\
  \hline
  $\mu_0$ & The number of the registered SUs\\
  \hline
  $\mu_1$ & The number of the unregistered SUs\\
  \hline
  $\phi_0(b)$ & The cost to reserve $b$ bandwidth\\
  \hline
  $\phi_1(q)$ & The maintenance cost for $q$ database queries\\
  \hline
  $T$ & The number of different SUs' types\\
  \hline
  $\theta^i$ & The $i$'th smallest user type\\
  \hline
  $\beta^i$ & The percentage of SUs with type $\theta^i$\\
  \hline
  $\gamma^i$ & The percentage of SUs with type $\theta^i$ with registration scheme\\
  \hline
\end{tabular}
\end{table}

\subsection{Pricing Framework}
\label{subsec:framework}

The DO offers two different pricing schemes, the registration scheme and the service plan scheme, simultaneously.

In the registration scheme, the SUs are registered with a uniform registration fee $r$. The DO reserves a fraction of $B_R/B$ ($0\leq B_R \leq B$) of the total TV bandwidth for them in each period. The expected number of reserved channels is $B_R/b_0$. The DO further divides the reserved channels equally to serve each registered SU. To make the problem solvable, in this paper, we consider only uniform registration fee. In a more flexible setting, the registered SUs can pay different amount of registration fees and enjoy different shares of the reserved bandwidth.
We also assume an SU can refuse to pay the registration fee to abort the deal.

In the service plan scheme, the DO offers the SUs several query plans, which are a set of $K+1$ query-price combinations denoted as $\mathcal{P}=\{(q_i, p_i), i=0,1,\cdots, K\}$. During the $M$ periods, the SUs with chosen plan $(q_i, p_i)$ can access the database $q_i$ times with a wholesale price $p_i$.
Note that $q_i$s are integers and $0\leq q_i\leq M$.
We assume there is a plan $(q_0=0, p_0=0)$ for SUs choosing neither scheme. The unregistered SUs shared the unreserved band in a time division manner. For example, these SUs can leverage CSMA in the MAC layer for channel access as today's Wi-Fi.

The pricing procedure is conducted at the beginning of the $M$ periods and can be viewed as two stages. At stage I, the DO announces parameters $B_R$ and $r$ for the registration scheme. Then the SUs decide whether to register. At stage II, the DO first announces $\mathcal{P}$ for all unregistered SUs and then these SUs specify which particular query plan to buy. An SU can always choose $(q_0, p_0)$ to abort the deal.

The procedure of pricing and interactions between the DO and the SUs are summarized in Fig. \ref{fig:procedure}.

\begin{figure}[tp]
\centering
  \includegraphics[width=3.2in]{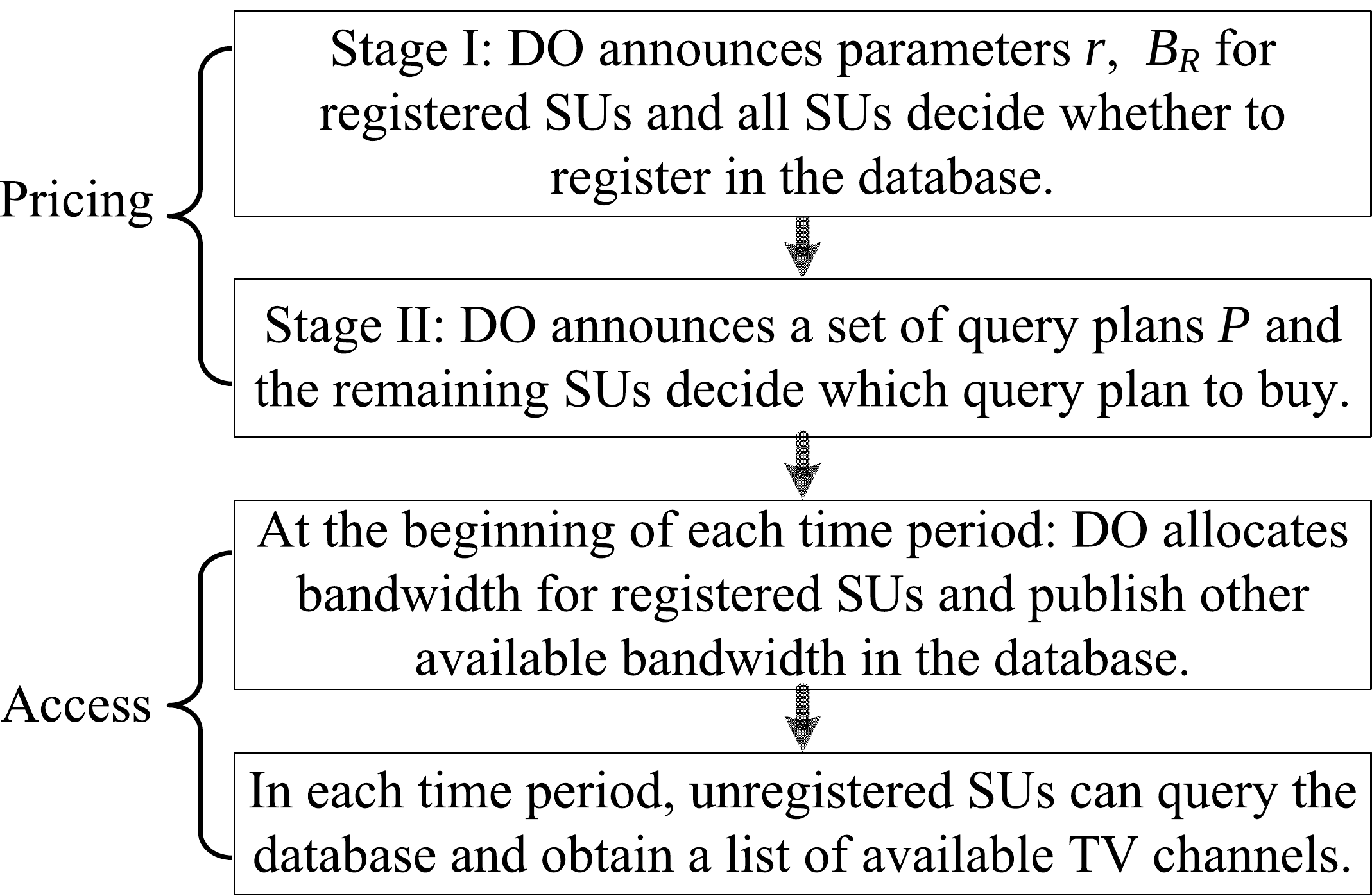}\\
  \caption{The procedure of pricing and access of TVWS.}
  \label{fig:procedure}
  \vspace{-0.0cm}
\end{figure}

\subsection{DO Model}
\label{subsec:DO}

We assume the DO has two different sources of database maintenance cost. On one hand, to reserve bandwidth for the registered SUs, DO needs to pay a bandwidth reservation fee to the regulators. We assume the cost to reserve $b$ MHz bandwidth is in the form $\phi_0(b)=\epsilon_0\cdot b^{\alpha}$, $\alpha\geq 1$. The reservation cost is a convex function on $b$. The key reason is that spectrum bandwidth is a limited resource. Reservation must be approved by the regulator and there can be multiple databases trying to reserve bandwidth. When the total demand is higher, the equilibrium price can be also higher \cite{reservation_price}. Furthermore, to avoid unnecessary reservations, the regulator may consider charge higher for more reserved bandwidth \cite{ReserveOfcom}. On the other hand, there is also the cost for the DO to calculate the available bandwidth and respond to the queries of unregistered SUs. Here we adopt a linear cost function on the query number. We assume for $q$ queries, the DO has to pay a marginal maintenance cost $\phi_1(q)=\epsilon_1\cdot q$.

We define the \emph{utility} of the DO, denoted by $U_{DO}$ as the difference between the payment received from all the SUs and the total cost. For convenience, we assume the number of SUs choosing query plan $i$ is $N_i, (i=0, 1, \cdots, K)$, the number of registered SUs is $\mu_0$ and the number of unregistered users is $\mu_1$ then
\begin{eqnarray}
\label{eqn:u_do}
U_{DO}&=&\mu_0\cdot r+\sum_{i=0}^K N_i\cdot p_i-\phi_0(B_R)-\phi_1(\sum_{i=0}^K N_i\cdot q_i).
\end{eqnarray}

\subsection{SU Model}
\label{subsec:SU}

According to the pricing framework, the SUs choosing the same pricing scheme share the spectrum evenly. For $u_i$, if there are $X$ SUs sharing the same bandwidth of $W$, its capacity is given by the Shannon-Hartley theorem as:
$
C_i(X, W)=\frac{1}{X}\cdot W\log_2\left(1+\frac{S_i}{n_0}\right)
$, where $S_i$ is the average received signal power over bandwidth $W$ at $u_i$'s receiver, $n_0$ is the average noise and interference power over the bandwidth. $S_i$ relates to the transmission power of $u_i$ and the path loss between $u_i$'s transmitter and receiver. Without loss of generality, we assume $n_0$ is identical for all SUs.

We assume SUs have linear valuations for any expected capacity they achieve. We define the valuation function of $u_i$ as:
$
V_i(X, W)=w_iC_i(X, W)
$, where $w_i$ is the valuation parameter of $u_i$. $w_i>0$ and its physical meaning is that the valuation of the unit channel capacity contributes to the overall benefit. $w_i$ relates to the personal information of $u_i$ such as the possible transmission duration in each period.

We define the \emph{utility} of $u_i$, denoted as $U_i$, as the difference between its valuation of the guaranteed capacity and the price charged by the DO. If $u_i$ chooses the registration scheme, it can occupy part of the TV bandwidth by itself. Its utility is
\begin{equation*}
U_i=V_i(1, \frac{B_R}{\mu_0})-r=w_i\cdot \frac{B_R}{\mu_0}\log_2\left(1+\frac{S_i}{n_0}\right)-r.
\end{equation*}
If $u_i$ chooses the service plan scheme, it should share the bandwidth with other unregistered SUs. There are at most $\mu_1$ unregistered SUs sharing, any one of them is guaranteed with a valuation of $V_i(\mu_1, B-B_R)$ in one access.  Moreover, $u_i$ can only access at most $q_j$ periods after selecting query plan $(q_j, p_j)$. Therefore, its utility is
\begin{eqnarray*}
U_i&=&V_i(\mu_1, B-B_R)\cdot\frac{q_j}{M}-p_j \\
&=&w_i\cdot \frac{B-B_R}{\mu_1}\log_2\left(1+\frac{S_i}{n_0}\right)\cdot \frac{q_j}{M}-p_j.
\end{eqnarray*}

From the model, we can see that $S_i$, $w_i$ are personal parameters for $u_i$. We further define $\theta_i=w_i\log_2\left(1+\frac{S_i}{n_0}\right)$ as the \emph{type} of $u_i$ and $v(q)=\frac{B-B_R}{\mu_1}\frac{1}{M}q$ as the \emph{valuation} of $q$ queries. The physical meaning of $\theta_i$ is the capacity of $u_i$ that can turn unit bandwidth into revenue. We use $a_i$ to denote the pricing scheme selected by $u_i$. $a_i=0$ represents the registration scheme and $a_i=1$ represents the service plan scheme. The utility function of $u_i$ can be summarized as
\begin{equation}
\label{equ:u_su}
U_i=\left\{
\begin{array}{ll}
\frac{B_R}{\mu_0}\theta_i-r, & \textrm{if $a_i=0$}  \\
\theta_iv(q_j)-p_j, & \textrm{if $a_i=1$ and the plan is $(q_j, p_j)$}
\end{array}
\right.
\end{equation}

We assume there are a total of $T$ different types and we denote the types by the set $\Theta=\{\theta^1, \theta^2, \cdots, \theta^T\}$. Without loss of generality, we assume $0<\theta^1< \theta^2< \cdots <\theta^T$.

\section{Problem Formulation}
\label{sec:formulation}

In this section, we formulate the pricing problem according to the two-stage pricing framework.

\subsection{SU Strategy}

In this paper, we consider two different cases of SUs based on their abilities to act strategically.
\begin{itemize}
  \item \emph{Non-strategic SUs}: in this case, for $u_i$, $a_i$ is fixed before the pricing procedure. SUs are categorized with pricing scheme choices. An SU will not change the choice of the pricing scheme strategically according to the choices of other users.
  \item \emph{Strategic SUs}: in this case, for $u_i$, $a_i$ is not fixed before the pricing procedure. $u_i$ can determine its pricing scheme dynamically according to the choices of other users.
\end{itemize}

Actually, we argue that both assumptions are reasonable in real world. In matured services \cite{icdcs12segmentation}, different pricing schemes have clear pros and cons. Users tend to have their preferences. In emerging services, the SUs may have no experiences or preferences.

In the case of non-strategic SUs, $a_i$ is fixed. We assume the fraction of SUs with type $\theta^i$ choosing $a_i=0$ is $\gamma^i$ ($i=1,\cdots, T$). And we assume the ratio of SUs of type $\theta^i$ among all SUs is $\beta^i$. Therefore, the total number of SUs with $a_i=0$ is $N\sum_{i=1}^T\beta^i\gamma^i$ and the total number of SUs with $a_i=1$ is $N\sum_{i=1}^T\beta^i(1-\gamma^i)$.

In the case of strategic SUs, $a_i$ is not fixed. SUs strategically choose either pricing scheme considering other SUs' strategies.
Therefore, at stage I, given $B_R$ and $r$, the SUs compete for pricing schemes. The interactions among the SUs form a \emph{Non-cooperative Database Registration Game (NDRG)} $G=(\mathcal{N}, \mathcal{S}=\{1,0\}, \{U_i\}_{u_i\in \mathcal{N}})$, where $\mathcal{N}$ is the set of players, $\mathcal{S}$ is the set of strategies, and $U_i$ is the payoff function of play $u_i$. Let $\mathbf{a}=(a_1, a_2, \cdots, a_N)$ be the strategy-tuple of all the SUs. Let $a_{-i}=(a_1, a_2, \cdots, a_{i-1}, a_{i+1}, \cdots, a_N)$ be the strategy-tuple of all the SUs expect $u_i$.
\begin{definition}[\textbf{Nash Equilibrium (NE) of NDRG}]
A strategy vector $\mathbf{a^*}=(a_1^*, a_2^*, \cdots, a_N^*)$ is an NE of the game NDRG if
$
a_i^*=\argmax_{a_i\in S}U_i(a_i, a_{-i}^*), \forall u_i \in\mathcal{N}.
$
\end{definition}

\subsection{Information Scenario}

In this paper, we also study the following two different information scenarios.
\begin{itemize}
  \item \emph{Complete information}: in this case, the DO and the SUs are perfectly informed about the type of each SU.
  \item \emph{Incomplete information}: in this case, the DO and the SUs know only the distribution of the SUs' types.
\end{itemize}

In the complete information scenario, the DO can treat each SU separately and offer a type-specific contract item. However, in the incomplete information scenario, the DO cannot observe the type of each SU, it has to offer the same set of query plans to all SUs.

To tackle the two different information scenarios, we model the design of query plans with contract theory \cite{contract_book}. The set of query plans $\mathcal{P}$ can be viewed as the set of contact items. The DO is the seller and the SUs are the buyers. The goods offered by the seller are the database queries and the wholesale price for $q_i$ queries is $p_i$. The buyers' types are their private information. According to the revelation principle \cite{contract_book}\cite{revelation}, it is sufficient to design at most $T$ contract items, one for each type of the SUs, to enable them to truthfully reveal their types. Therefore, we assume $K=T$ and an unregistered SU of type $\theta^i$ will eventually choose plan $(q_i, p_i)$.

\subsection{Backward Induction for the Two-stage Pricing Framework}

Based on the two-stage pricing procedure, the design of the optimal pricing schemes can be analyzed with backward induction. More specifically, we first show how should the DO design the query plans for the SUs at stage II, given $B_R$, $r$ and SUs' strategy-tuple $\mathbf{a}$. Then we study how the SUs choose the pricing scheme given $B_R$, $r$ announced at stage I. Finally, based on the knowledge of the possible pricing scheme selection of the SUs, the DO optimally selects $B_R$ and $r$ at stage I.

In this paper, we will first consider a baseline scenario where SUs are non-strategic players and the DO processes complete information of all the SUs. Then, we extend our analysis to strategic SUs under both complete and incomplete information scenarios. By comparing the first and the second scenario, we show the impact of SU strategy on the pricing design. By comparing the second and the third scenario, we show the impact of information completeness. Due to the page limit and similarity in analysis, we omit the case of non-strategic SUs under incomplete information.

\section{Pricing Solution for Non-strategic SUs in Complete Information Scenario}
\label{sec:non-strategic_complete}

In this section, we consider a baseline case where the SUs are non-strategic players and the DO has complete information of all the SUs.

\subsection{SU Behavior Given Pricing Parameters}

Since the SUs fix their choices of pricing schemes, we can analyze the registered and unregistered SUs separately.

At stage I, some SUs with $a_i=0$ may refuse to pay the fee if it is too high for them. Given $B_R$ and $r$, $u_i$ pays the registration fee if and only if the following condition holds:
$
\frac{B_R}{\mu_0^*}\theta_i-r\geq 0,
$
where $\mu_0^*=N\sum_{i=1}^T\beta^i\gamma^i$ is the prior knowledge of the number of registered SUs. Non-strategic SUs decide whether to pay the fee only based on the prior knowledge.
The actual number of registered SUs can then be calculated as
\begin{equation}
\label{eqn:mu0}
\mu_0=\sum_{i\in\left\{i|\frac{B_R}{\mu_0^*}\theta^i-r\geq 0\right\}}N\beta^i\gamma^i.
\end{equation}

At stage II, there are $\mu_1=N\sum_{i=1}^T\beta^i(1-\gamma^i)$ unregistered SUs. Since any $u_i$ can be guaranteed with non-negative utility by choosing $(q_0, p_0)$, to enable $u_i$ to select query plan $(q_i, p_i)$, $q_i, p_i$ should satisfy $\theta_iv(q_i)-p_i\geq 0, \ \forall 1\leq i\leq N$.
The constraints are usually referred to as the \emph{Individual Rationality} (IR).
The following lemma says that all the IR constraints are bind under complete information.
\begin{lemma}
\label{lemma:IR_complete}
At stage II, if the DO has complete information of $\theta_i$ for each $u_i$, the optimal query plan $(q_i, p_i)$ offered to $u_i$ satisfies: $\theta_iv(q_i)-p_i=0, \ 1\leq i\leq N$.
\end{lemma}
\begin{proof}
We prove by contradiction. Suppose in the optimal pricing design, the DO offers $(q_i^*, p_i^*)$ to $u_i$ at stage II to optimize $U_{DO}$ and $\theta_iv(q_i^*)-p_i^*>0$. However, when the DO increases the price $p_i^*$ by $\Delta p_i^*=\theta_iv(q_i)^*-p_i^*$ and keeps all the other pricing parameters unchanged. $U_{DO}$ can be increased by $\Delta p_i^*$, which contradicts with the optimality of $(q_i^*, p_i^*)$. Therefore, we must have $\theta_iv(q_i^*)-p_i^*=0$.
\end{proof}
According to Lemma \ref{lemma:IR_complete}, we have:
\begin{equation}
\label{eqn:IR}
\theta_iv(q_i)-p_i = 0, \quad \forall 1\leq i\leq N.
\end{equation}

\subsection{Optimal Parameter Selection of DO}

By substituting (\ref{eqn:mu0}) and (\ref{eqn:IR}) into (\ref{eqn:u_do}) we can rewrite $U_{DO}$ as
\begin{eqnarray}
  \label{eqn:UDO0}
U_{DO}(B_R, r, \mathcal{P})&=&\mu_0\cdot r + \sum_{i=0}^KN_ip_i-\phi_0(B_R)-\phi_1(\sum_{i=0}^KN_iq_i) \nonumber \\
&=&\sum_{i\in\left\{i|\frac{B_R}{\mu_0^*}\theta^i-r\geq 0\right\}}N\beta^i\gamma^i\cdot r - \epsilon_0B_R^{\alpha} \nonumber\\
&+&\sum_{i=1}^TN\beta^i\left(1-\gamma^i\right)\left(\frac{B-B_R}{\mu_1M}\theta^i-\epsilon_1\right)\cdot q_i.
\end{eqnarray}

The optimization problem for DO is
\begin{eqnarray}
\label{eqn:max_udo0}
&&\max_{B_R, r, \{q_i\}} U_{DO}(B_R, r, \mathcal{P}) \nonumber\\
&\textrm{subject to:} & 0\leq q_i\leq M (\forall 1\leq i\leq T), \quad\textrm{and}\quad 0\leq B_R\leq B, \nonumber\\
&\textrm{variables:} & \{q_i(\forall 1\leq i\leq T), B_R, r\}.
\end{eqnarray}

Problem (\ref{eqn:max_udo0}) is a non-convex optimization problem and there is no closed-form solution.
The first term $\sum_{i\in\left\{i|\frac{B_R}{\mu_0^*}\theta^i-r\geq 0\right\}}N\beta^i\gamma^i$ in Eqn. (\ref{eqn:UDO0}) is a piecewise constant function of $B_R$. The $M$ points $\frac{r}{\theta^{i}}\mu_0^*$ divide the domain of $B_R$ into $M+1$ different intervals. It is a quadratic programming problem to optimize (\ref{eqn:UDO0}) in each interval. The problem (\ref{eqn:max_udo0}) thus can be decomposed into $M+1$ quadratic programming problems, which may be NP-hard in the general form \cite{nonlinear}.

Considering the nature of our problem, the physical meaning of $B_R$ is the bandwidth reserved for registered SUs and $B_R$ is measured in channels. The possible values of $B_R/b_0$ are integers in the range $[0, B/b_0]$. Also, we observe that $U_{DO}$ has a linear relationship for all $\{q_i\}$, then $q_i$ will either be 0 or $M$.
\begin{equation}
\label{eqn:q0}
q_i=\left\{
\begin{array}{ll}
0 & \textrm{if $\frac{B-B_R}{\mu_1}\frac{1}{M}\theta^i-\epsilon_1\leq0$} \\
M & \textrm{otherwise}
\end{array}
\right.
\end{equation}
Furthermore, the range and scale of the registration price $r$ are also limited in reality. If we also restrict the range of $r$ to be $[\emph{\b{r}}, \emph{\={r}}]$ and its minimum scale to be $r_0$, the DO needs only to explore a total of $\frac{B}{B_R}\cdot\frac{\emph{\={r}}-\emph{\b{r}}}{r_0}$ combinations to find the solution.
In this paper, we will leverage this simpler but efficient two-dimensional search method to find the pricing parameters.

Therefore DO tries to solve the following optimization problem
\begin{eqnarray}
\label{eqn:max_udo1}
&&\max_{B_R, r, \{q_i\}} U_{DO}(B_R, r, \mathcal{P}) \nonumber\\
&\textrm{subject to:}& \quad (\ref{eqn:q0}) \quad\textrm{and}\quad 0\leq B_R\leq B.
\end{eqnarray}
We provide the numerical results in Section \ref{sec:simulation}.

\section{Pricing Solution for Strategic SUs in Complete Information Scenario}
\label{sec:strategic_complete}

In this section, we consider the SUs to be strategic players.

\subsection{Stage II: Optimal Contract Design under Complete Information}

Same as the case discussed in Section \ref{sec:non-strategic_complete}, when the DO processes complete information of the SUs, any SU choosing the service plan scheme will have zero utility.
The contract items ${(q_i, p_i)}$ satisfy the condition in Eqn. (\ref{eqn:IR})

\subsection{Stage I: Database Registration Game of SUs}

\subsubsection{Existence of an NE}

At stage I, SUs compete against each other in the NDRG considering the outcome at stage II given the parameters $B_R$ and $r$. The utility function of $u_i$ can be rewritten as
\begin{eqnarray}
U_i(a_i, a_{-i})=\left\{
\begin{array}{ll}
\frac{B_R}{\mu_0}\theta_i-r & \textrm{if $a_i=0$} \\
0 & \textrm{if $a_i=1$}
\end{array}
\right. \label{eqn:u_complete_info}
\end{eqnarray}

We show that the NDRG under complete information is a \emph{Unweighted Congestion Game} \cite{congestion_game}.
\begin{definition}[\textbf{Unweighted Congestion Game}]
A non-cooperative game satisfying the following condition is referred to as the unweighted congestion game:
\begin{itemize}
  \item The players share a common set of strategies.
  \item The payoff the $i^{th}$ player receives for playing the $j^{th}$ strategy is a monotonically non-increasing function of the total number $\mu_j$ of players playing the $j$th strategy.
\end{itemize}
\end{definition}

\begin{lemma}
\label{lemma:NDRG_is_CG}
NDRG under complete information is an unweighted congestion game.
\end{lemma}
\begin{proof}
It is easy to see that all the players in NSDG share a common set of strategies $\mathcal{S}$. In the scenario of complete information, the payoff function of player $u_i$ is defined as $U_i$ in (\ref{eqn:u_complete_info}). For strategy 1, $U_i$ is constant. For strategy 0, $U_i$ is a monotonically non-increasing function with $\mu_0$.
\end{proof}

The nice properties of the unweighted congestion game are summarized as lemma \cite{congestion_game}.
\begin{lemma}
\label{lemma:property_of_CG}
Every unweighted congestion game $\Gamma$ possesses an NE in pure strategies. Given an arbitrary strategy-tuple $\mathbf{a}$ in $\Gamma$, there exists a best-reply improvement path \cite{congestion_game} with $L$ steps and $L\leq s{{n+1}\choose 2}$, where $s$ is the number of all strategies and $n$ is the number of players in $\Gamma$.
\end{lemma}

\subsubsection{Distributed Database Registration Algorithm}

Lemma \ref{lemma:property_of_CG} confirms the existence of an NE in NDRG under complete information. Starting from any strategy-tuple, an NE can be achieved with $L\leq s{{n+1}\choose 2}$ steps. In our scenario, $s=2$ and $n=N$. $L$ is therefore bounded by $N(N+1)$. Observe that when selecting the registration scheme, SUs of higher types have higher utility. When more SUs choosing the registration scheme, the utility of SUs of lower types becomes non-positive sooner. Therefore, the best strategy for an SU is to make sure that no more SUs with higher types will choose the registration scheme. Otherwise, if an SU chooses registration and more SUs with higher types also choose, the SU will be under the risk of non-positive utility. We can then design a distributed algorithm for SUs to achieve an NE without any improvement steps. The key idea is to allow the SUs to choose the registration scheme in a descending order of their types until no more SUs have incentives to join.

We assume that each SU is assigned a unique ID by the DO when connecting with the database and it knows the number of existing registered SUs at stage I, by checking the database. Without loss of generality, we assume the ID of $u_i$ is $i$. In the complete information scenario, the SUs know the exact value of the type and the ID of each other. Therefore, all the SUs can independently calculate a threshold of user type denoted as $\theta_{i_M}$. When SUs of type higher than $\theta_{i_M}$ are registered, their utilities are positive. But when SUs with type equal to $\theta{i_M}$ also register, these SUs can only have non-positive utility. Therefore, only part of the SUs with type equal to $\theta{i_M}$ can choose registration. After the calculation, the SUs with type higher than $\theta_{i_M}$ then choose the registration scheme. And the SUs of $\theta_{i_M}$ choose to register in the order of the assigned ID. Finally the remaining SUs with even lower types choose the service plan scheme. The decisions of the SUs are made at the same time.
The procedure is summarized in Algorithm \ref{alg:registration_complete_info}.

\begin{algorithm}[tp]
\caption {Distributed Database Registration Algorithm.}
\label{alg:registration_complete_info}
\begin{algorithmic}[1]
\STATE $a_i=1, \forall 1\leq i\leq N$; $\quad i_M=T$; $\quad R=0$;
\WHILE {$B_R\theta^{i_M}>(R+N\cdot\beta^{i_M})\cdot r$}
\STATE $R=R+N\cdot\beta^{i_M}$; $\quad i_M=i_M-1$;
\STATE \textbf{if} {$i_M=0$} \textbf{then} BREAK out of WHILE;
\ENDWHILE
\STATE $a_i=0, \forall \theta_i> \theta^{i_M}$;
\FOR {$i \in \{1, \cdots, N\}$}
\IF {{$\theta_i=\theta^{i_M}$} \AND {$B_R\theta_{i}>(R+1)\cdot r$}}
\STATE $R=R+1$; $\quad a_i=0$;
\ENDIF
\ENDFOR
\end{algorithmic}
\end{algorithm}

\begin{theorem}
\label{thm:algorithm1_property}
Under complete information, Algorithm \ref{alg:registration_complete_info} is guaranteed to achieve an NE of NDRG. Its time complexity is $O(N)$.
\end{theorem}
\begin{proof}
We first prove that all the SUs with the registration scheme have positive payoffs. In Algorithm \ref{alg:registration_complete_info}, suppose the last SU selecting the registration scheme is $u_i$, $u_i$ has positive utility since $B_R\theta_{i}>\mu_0\cdot r$ ($\mu_0$ is the final number of registered SUs). Any other registered SU $u_j$, must has $\theta_{j}\geq \theta_{i}$. Therefore, we have $B_R\theta_{j}\geq B_R\theta_{i}>\mu_0\cdot r$.

Then we show that there is no SU with the service plan scheme has the incentive to change its strategy. Any unregistered SU $u_k$ with $\theta_k=\theta_i$ will not change strategy since the condition $B_R\theta_{k}>(\mu_0+1)\cdot r$ is not satisfied. Therefore, its payoff will not improve by changing to the registration scheme. Also, for any other unregistered SU $u_l$ with $\theta_l<\theta_k$, $B_R\theta_{l}\leq B_R\theta_{k}\leq(\mu_0+1)\cdot r$.

In summary, no player has the incentive to deviate from the strategy provided by Algorithm \ref{alg:registration_complete_info}. It is also easy to see the number of operations in the algorithm is $O(N)$.
\end{proof}

Each SU can independently implement Algorithm \ref{alg:registration_complete_info} by acquiring the number of currently registered SUs. Note that the converged NE is not unique if there are multiple SUs with the same type $\theta^{i_M}$. Due to the NE output proved in Theorem \ref{thm:algorithm1_property}, no SU has the incentive to deviate unilaterally from the algorithm.

\subsubsection{Stage I: Optimal Parameter Selection of DO}

At stage I, DO selects pricing parameters considering the NE of NDRG. From Algorithm \ref{alg:registration_complete_info}, DO can get an estimation of $\mu_0$ and $\mu_1$ as
$
\mu_0=\mathcal{R}(B_R, r)=R_0+R_1, \quad\mu_1=N-\mu_0,
$
where $R_0=\sum_{i=i_M+1}^T N\cdot\beta^i$, $R_1=\arg\max_k\left\{\frac{B_R\theta^{i_M}}{r}>R_0+k\right\}$ and $i_M=\arg\min_k\left\{\frac{B_R\theta^k}{r}>\sum_{i=k}^T N\cdot\beta^i\right\}-1$.
$R_0$ is the total number of SUs with type larger than $\theta^{i_M}$. $R_1$ is the number of registered SUs of type $\theta^{i_M}$.

Together with Eqn. (\ref{eqn:IR}), we can rewrite $U_{DO}$ in (\ref{eqn:u_do}) as:
\begin{eqnarray}
U_{DO}(B_R, r, \mathcal{P})&=&\left(\mu_0r-\epsilon_0B_R^{\alpha}\right) + \sum_{i=1}^{i_M-1} N\beta^i\left(\frac{B-B_R}{\mu_1M}\theta^i-\epsilon_1\right) q_i \nonumber\\
&+&\left(N\beta^{i_M}-R_1\right)\left(\frac{B-B_R}{\mu_1M}\theta^{i_M}-\epsilon_1\right) q_{i_M}.
\end{eqnarray}

Here $U_{DO}$ also has a linear relationship for all $\{q_i\},\forall i<i_M$, we have:
\begin{equation}
\label{eqn:q}
q_i=\left\{
\begin{array}{ll}
0 & \textrm{if $\frac{B-B_R}{\mu_1}\frac{1}{M}\theta^i-\epsilon_1\leq0$} \\
M & \textrm{otherwise}
\end{array}
\right.
\end{equation}
The optimization problem for the DO becomes
\begin{eqnarray}
\label{eqn:max_udo2}
&&\max_{B_R, r, \{q_i\}} U_{DO}(B_R, r, \mathcal{P}) \nonumber\\
&\textrm{subject to:}& \quad (\ref{eqn:q}) \quad\textrm{and}\quad 0\leq B_R\leq B.
\end{eqnarray}
Similarly, we restrict the range of $r$ to be [\emph{\b{r}}, \emph{\={r}}] and use a two-dimensional exhaustive search to find the solution. We provide the numerical results in Section \ref{sec:simulation}.

\section{Pricing Solution for Strategic SUs in Incomplete Information Scenario}
\label{sec:strategic_incomplete}

In this section, we study the pricing problem under incomplete information for strategic SUs. Due to hidden information, at stage II, the DO cannot design query plans for each SU to extract all their revenue. Instead, the DO should offer a set of contract items for each type of SUs to choose from. Since SUs can obtain non-negative utility choosing the service plan scheme, the analysis of the NDRG will be much more complicated.

\subsection{Stage II: Contract Design under Incomplete Information}

In the incomplete information scenario, to ensure SUs of type $\theta^i$ have the incentive to select query plan $(q_i, p_i)$, another set of constraints called \emph{Incentive Compatibility} (IC) should be satisfied:
\begin{equation}
\theta^iv(q_i)-p_i\geq \theta^iv(q_j)-p_j, \forall 1\leq i,j\leq T.
\end{equation}
The physical meaning of IC is that an SU with type $\theta^i$ achieves the maximum utility when choosing the corresponding contract item $(q_i, p_i)$. Note that there are a total of $N(N-1)$ IC constraints.

To analyze the design of optimal contract items, we first relax the condition that $q_i (1\leq i\leq T)$ are integer numbers and allow $q_i$s to be any real number in the range $[0, M]$.
With the relaxed $q_i$s, the optimal contract design problem satisfies the following condition.
\begin{definition}[\textbf{Spence-Mirrlees Single-crossing Condition (SMC)} \cite{contract_book}]
Spence-Mirrlees single-crossing condition (SMC) is satisfied if the user's utility function $U(q, p, \theta)$ satisfies:$\frac{\partial}{\partial\theta}\left[-\frac{\partial U/\partial q}{\partial U/\partial p}\right]>0$, where $q$ is quantity of items, $p$ is the price for $q$ items and $\theta$ is the type of user.
\end{definition}
When SMC is satisfied, the $N(N-1)$ IC constraints can be reduced to a set of tractable equivalent constraints \cite{contract_book}.
\begin{lemma}
\label{lemma:IC_incomplete}
If SMC is satisfied, the necessary and sufficient conditions for the satisfaction of all the IC constraints are:
\begin{eqnarray}
\label{eqn:monotone_q}
\left\{
\begin{array}{l}
\theta^1v(q_1)-p_1=0 \\
\theta^iv(q_i)-p_i=\theta^iv(q_{i-1})-p_{i-1}, \quad 1 < i \leq T\\
q_i\geq q_j \quad \textrm{where} \quad \theta^i\geq\theta^j
\end{array}
\right.
\end{eqnarray}
\end{lemma}

It is easy to verify that the SMC condition holds when $q_i$s are real numbers.
Based on Lemma \ref{lemma:IC_incomplete}, introducing a $\theta^0=0$, we can express the query plans as
\begin{eqnarray}
\label{eqn:pk}
\left\{
\begin{array}{l}
p_k=\sum_{i=1}^k \left[\theta^iv(q_i)-\theta^iv(q_{i-1})\right],\quad \forall 1\leq k\leq T\\
q_i\geq q_j \quad \forall i\geq j
\end{array}
\right.
\end{eqnarray}
Substituting (\ref{eqn:pk}) and $N_i=\mu_1\cdot\beta^i$ into (\ref{eqn:u_do}) and putting the terms related to the same query variable together, we have:
\begin{eqnarray*}
U_{DO}(\mathcal{P}) &=&(\mu_0r-\epsilon_0B_R^{\alpha}) \\
&+& \sum_{i=1}^T\left[
\mu_1\beta^i\theta^iv(q_i)-\Delta^iv(q_i)\sum_{j=i+1}^T\mu_1\beta^j - \mu_1\beta^i\epsilon_1q_i\right] \\
&=&(\mu_0r-\epsilon_0B_R^{\alpha}) + \sum_{i=1}^T g_i(\mu_1)q_i,
\end{eqnarray*}
where $\Delta^i=\theta^{i+1}-\theta^i,\forall i<T$, $\Delta^T=0$ and $g_i(\mu_1)=\beta^i\theta^i\frac{B-B_R}{M}-\Delta^i\frac{B-B_R}{M}\sum_{j=i+1}^T\beta^j - \mu_1\beta^i\epsilon_1$.
The optimal contract design problem for DO can then be expressed as
\begin{eqnarray}
&U_{DO}(\mathcal{P})&=(\mu_0r-\epsilon_0B_R^{\alpha}) + \sum_{i=1}^Tg_i(\mu_1)q_i \nonumber\\
&\max_{\{q_i\}}& U_{DO}(\mathcal{P}) \nonumber\\
&\textrm{subject to:}& 0\leq q_i\leq  M, \quad q_i\geq q_j,\,\forall i\geq j.
\end{eqnarray}

\begin{algorithm}[tp]
\caption {Finding Valid Query Sequence.}
\label{alg:valid_q}
\begin{algorithmic}[1]
\STATE $i_S(\mu_1)=T$;
\WHILE {$g_{i_S(\mu_1)}(\mu_1)>0$}
\STATE $i_S(\mu_1)=i_S(\mu_1)-1$;
\ENDWHILE
\STATE $\hat{q_i}(\mu_1)=0, \forall i\leq i_S(\mu_1)$; $\quad \hat{q_i}(\mu_1)=M, \forall i> i_S(\mu_1)$;
\end{algorithmic}
\end{algorithm}

Now, we consider $q_i$s to be integers. Since factor $g_iq_i$ only relates to $q_i$ given $\mu_1$, we can then find the optimal value for each $q_i$ independently. We use $\tilde{q_i}(\mu_1)$ to denote $\argmax_{0\leq q_i\leq  M} g_iq_i$. The set of $\{\tilde{q_i}(\mu_1)\}$ can be achieved at the boundary points (0 or $M$). It is worth noting that $\{\tilde{q_i}(\mu_1)\}$ is a function of $\mu_1$ since $g_i$ is a function of $\mu_1$. However, $\tilde{q_i}$ may not be monotonically nondecreasing on $i$ (requirement of Eqn. (\ref{eqn:monotone_q})). For example, if for some $i^*$, $g_{i^*}>0$ and $g_{i^*+1}<0$ then $q_{i^*}=M$ and $q_{i^*+1}=0$. We design Algorithm \ref{alg:valid_q} to obtain valid $\{q_i\}$. We first find the largest $i_S(\mu_1)$, such that $g_{i_S(\mu_1)}<0$ and set all $q_{j}, j\leq i_S(\mu_1)$ to be 0. Note that $q_i$s obtained in Algorithm \ref{alg:valid_q} can be 0 or $M$, which is the same as that when $q_i$s are real numbers. Therefore, it is valid for us to replace the IC constraints with Eqn. \ref{eqn:monotone_q}.

Suppose the valid queries obtained are $\hat{q_i}$ for type $\theta^i$ SUs, the optimal query plans at stage II are
\begin{eqnarray*}
\left\{
\begin{array}{lll}
\mathcal{P}&=&\{(q_i, p_i)\}\bigcup\{(q_0=0, p_0=0)\},\quad 1\leq i \leq T \\
q_i&=&\hat{q_i}(\mu_1)\\
p_i&=&\hat{p_i}(\mu_1)=\sum_{j=1}^i \left[\theta^jv(\hat{q}_j(\mu_1))-\theta^jv(\hat{q}_{j-1}(\mu_1))\right]
\end{array}
\right.
\end{eqnarray*}

Note that Algorithm \ref{alg:valid_q} finds a suboptimal sequence of valid $\{q_i\}$. Because Algorithm \ref{alg:valid_q} has restricted the range of $\{q_i\}$ to be only boundary points. One way to obtain optimal $\{q_i\}$ is the brunching and ironing algorithm \cite{contract_book}. Instead of restricting the range of $\{q_i\}$, this algorithm first finds the optimal sequence of $\{q_i\}$ ignoring the condition of Eqn. (\ref{eqn:monotone_q}). Then it continues to find pairs of $q_i, q_{i+1}$ with $q_i<q_{i+1}$. After that, it finds a new $q_i^*\in[q_i, q_{i+1}]$ such that the payoff is maximized when setting $q_i=q_{i+1}=q_i^*$. It continues until the condition of Eqn. (\ref{eqn:monotone_q}) is satisfied. The algorithm is named since it ``iron'' out decreasing sub-sequence in $\{q_i\}$.  However, the optimal algorithm is of high complexity and cannot guarantee the existence of an NE in the NDRG at stage I. It is difficult to prove theoretical the gap between the optimal and suboptimal solution in this paper, we will evaluate the price of the suboptimal $\{q_i\}$ in Section \ref{sec:simulation}.

\subsection{Stage I: Database Registration Game of SUs}

We now consider the NDRG between SUs under incomplete information scenario.

\subsubsection{Existence of an NE}

At stage II, Algorithm \ref{alg:valid_q} is published by the DO to the SUs. Assuming the type of $u_i$ is $\theta^{\sigma(i)}$, the utility of $u_i$ when choosing the service plan scheme can be computed as:
$
U_i=\theta^{\sigma(i)}v(\hat{q}_{\sigma(i)}(\mu_1))-\hat{p}_{\sigma(i)}(\mu_1)
$.
Note that $\mu_1$ is the number of SUs choosing the service plan scheme. We can therefore rewrite the utility of $u_i$ considering the two possible strategies as
\begin{eqnarray}
\label{eqn:u_imcomplete_info}
U_i(a_i, a_{-i})=\left\{
\begin{array}{ll}
\frac{B_R}{\mu_0}\theta_i-r & \textrm{if $a_i=0$} \\
\theta^{\sigma(i)}\frac{B-B_R}{\mu_1M}(\hat{q}_{\sigma(i)}(\mu_1))-\hat{p}_{\sigma(i)}(\mu_1) & \textrm{if $a_i=1$}
\end{array}
\right. \label{eqn:u_incomplete_info}
\end{eqnarray}
From the utility function of $u_i$, we can verify that NDRG under incomplete information is also a congestion game when using Algorithm \ref{alg:valid_q} to obtain $\{q_i\}$.
\begin{lemma}
\label{lemma:NDRG_is_CG2}
NDRG under incomplete information is an unweighted congestion game.
\end{lemma}
\begin{proof}
All the players in NSDG share a common set of strategies $\mathcal{S}$. The payoff function of player $u_i$ is defined as $U_i$ in (\ref{eqn:u_incomplete_info}). For strategy $0$, $U_i$ is a monotonically non-increasing function with $\mu_0$. For strategy $1$, $g_i(\mu_1)$ is a non-increasing function of $\mu_1$. Therefore, $i_S(\mu_1)$ is also a non-increasing function of $\mu_1$. As a result, $\hat{q_i}(\mu_1)$ is a non-increasing function of $\mu_1$.
When $\hat{q_i}(\mu_1)=0$, $U_i=0$ and when $\hat{q_i}(\mu_1)=M$, $U_i\geq0$. Therefore $U_i$ is also a monotonically non-increasing function with $\mu_1$.
\end{proof}
As a result, we can also find an NE with an improvement path with a finite number of improvement steps. According to Lemma \ref{lemma:property_of_CG} and Lemma \ref{lemma:NDRG_is_CG2}, we have:
\begin{theorem}
\label{thm:NDRG_property2}
Under incomplete information, NDRG has a pure strategy NE which can be achieved with an improvement path of at most $N(N+1)$ steps starting from any strategy-tuple.
\end{theorem}

\begin{algorithm}[tp]
\caption {Distributed Database Registration Algorithm under Incomplete Information.}
\label{alg:registration_incomplete_info}
\begin{algorithmic}[1]
\STATE SUs randomly select one pricing scheme and record in the database;
\STATE DO assigns each SU a unique ID in the range $[0, N - 1]$;
\STATE Time duration counter $t=0$;
\WHILE {$\mathbf{a}$ is not an NE}
\STATE Active SU ID in $t$ is $i = t\mod N$;
\IF {$U_i(a_i, \mu_{a_i})<U_i(1-a_i, \mu_{1-a_i}+1)$}
\STATE Update strategy from $a_i$ to $1-a_i$ within $t$;
\ENDIF
\STATE $t=t+1$;
\ENDWHILE
\end{algorithmic}
\end{algorithm}

\subsubsection{Distributed Database Registration Algorithm}

In the case of incomplete information, some SUs of higher types are guaranteed with non-negative payoff with the service plan scheme. These SUs may prefer the service plan scheme.

Unlike the complete information case, the SUs now have to make the pricing scheme choice repeatedly. The SUs can act according to Algorithm \ref{alg:registration_incomplete_info} to achieve an NE in a distributed manner via repeated improvements. Algorithm \ref{alg:registration_incomplete_info} is based on the improvement path. Since all the SUs are connected with the database, we assume at stage I, the SUs are synchronized with the database. Also, we assume there are predefined \emph{time slots}. SUs can determine whether to change their strategies within a time slot one by one according to the order of their assigned IDs. In each time slot, only one SU may update his strategy to optimize its own utility given the new situation of other SUs' choices. It is easy to see that in Algorithm \ref{alg:registration_incomplete_info}, every update is a step in the improvement path. Therefore, Algorithm \ref{alg:registration_incomplete_info} is guaranteed to find an NE within $N(N+1)$ updates. Note that the found NE by the algorithm is also not unique if there are many SUs with the same type.

\subsection{Stage I: Optimal Parameter Selection of DO}

\subsubsection{Estimate the SUs' choices in an NE}

Estimating $\mu_0$ and $\mu_1$ is essential for the DO to optimize its pricing parameters. However, due to the randomness of the SUs' strategies the NE found by Algorithm \ref{alg:registration_incomplete_info} is not unique. Therefore, it is impossible to get an accurate estimation. We assume the DO uses the expected value as an estimation of $\mu_0$ and $\mu_1$. However, to calculate the expectation, the DO needs to generate all possible strategy-tuples and check all possible NEs. Therefore the DO has to generate $2^N$ different strategy-tuples, which is computationally infeasible. However, by exploring the nature of SUs in our scenario, we are able to dramatically reduce the computation complexity of the DO. The following lemma provides a way for the DO to estimate $\mu_0$ and $\mu_1$ with linear complexity with $N$.

\begin{lemma}
\label{lemma:NDRG_property_incomplete}
For every NE strategy-tuple $\mathbf{a}$, there exists another NE $\mathbf{\bar{a}}$ with the same number of registered SUs. In $\mathbf{\bar{a}}$, there is only one $\theta^{i_S}, 1\leq i_S\leq T$, such that $\forall i\neq i_S, \forall u_j, u_k$, if $\,\theta^j=\theta^k=i, a_j=a_k$.
\end{lemma}
\begin{proof}
Suppose in an NE, two of the user types are $\theta^i$ and $\theta^j$ and the strategy tuple is $\mathbf{a}$. Suppose $n_i^0$ of the $\theta^i$ users are choosing $0$ and $n_i^1$ of the $\theta^i$ SUs are choosing $1$. $n_j^0$ of the $\theta^j$ SUs are choosing $0$ and $n_j^1$ of the $\theta^j$ SUs are choosing $1$. Assume one $\theta^i$ SU $u_i$ changes its strategy from $0$ to $1$ and one $\theta^j$ SU $u_j$ changes its strategy from $1$ to $0$ at the same time. After the above strategy inter-changing, the strategy tuple, denoted as $\mathbf{a'}$ is also an NE. That is because $\mu_0$ and $\mu_1$ are the same in both $\mathbf{a}$ and $\mathbf{a'}$.
And none of the other SUs except $u_i$ and $u_j$ will have the incentive to change strategy in the previous NE tuple $\mathbf{a}$. We can also see that $u_i$ and $u_j$ have no incentive to change strategies since no other $\theta^i$ and $\theta^j$ SUs want to change from $0$ to $1$ or from $1$ to $0$.

Similarly, another equivalent NE can be achieved if some of the type $\theta^i$ SUs change from $1$ to $0$ and some of the type $\theta^j$ SUs change from $0$ to $1$ at the same time. By interchanging strategy choices of SUs in an NE scenario, we can obtain an NE tuple $\mathbf{\bar{a}}$ in which there is at most one type of SUs, say, $i_S$, satisfying the conditions in Lemma \ref{lemma:NDRG_property_incomplete}.
\end{proof}

Lemma \ref{lemma:NDRG_property_incomplete} says all SUs with the same type other than the type $\theta^{i_S}$ have the same strategies in $\mathbf{\bar{a}}$.
As a result, the DO can use Algorithm \ref{alg:estimate_SU_incomplete} to estimate the number of $\mu_0$ and $\mu_1$. It first chooses one type to be the $\theta^{i_S}$ and generates all possible strategy combinations of the other types. Then the DO checks whether it is possible to be an NE under the chosen $\theta^{i_S}$ by varying the number of $\theta^{i_S}$ type SUs' choices on strategy $0$ and $1$. In the algorithm, the DO only needs to check no more than $2^T\cdot N$ different strategy-tuples which is linear with $N$.

\begin{algorithm}[tp]
\caption {Estimating the SUs' Strategy under Incomplete Information.}
\label{alg:estimate_SU_incomplete}
\begin{algorithmic}[1]
\STATE Multi-set $R_{ALL}=\{\}$; // all possible values of $\mu_0$
\FOR {Each user type $\theta^i$}
\FOR {$Counter$=1 to $2^T$}
\STATE $a_{-i}$= the bits in $Counter$'s binary representation;
\STATE Assign strategy of SUs with type $\theta\neq\theta^i$ with $a_{-i}$;
\FOR {$j=0$ to $N\cdot \beta^i$ }
\STATE Let $j$ of the type $\theta^i$ SUs choose $0$;
\STATE Let $N\cdot\beta^i-j$ of the type $\theta^i$ SUs choose $1$;
\STATE \textbf{if} {It is an NE} \textbf{then} $R_{ALL} = R_{ALL}\bigcup \{\textrm{current}\ \mu_0\}$;
\ENDFOR
\ENDFOR
\ENDFOR
\STATE $\mu_0=$ the most common value in the multi-set $R_{ALL}$;
\STATE $\mu_1=N-\mu_0$;
\end{algorithmic}
\end{algorithm}

\subsubsection{Optimization Parameter Selection}

Utilizing the estimation of $\mu_0$ and $\mu_1$ obtained from Algorithm \ref{alg:estimate_SU_incomplete}, the DO now needs to solve the following optimization problem:
\begin{eqnarray}
\label{eqn:max_udo3}
&U_{DO}(B_R, r)&=(\mu_0(B_R, r)r-\epsilon_0B_R^{\alpha}) + \sum_{i=1}^Tg_i(B_R, r)\hat{q}_i(B_R, r) \nonumber\\
&\max_{B_R, r}& U_{DO}(B_R, r) \nonumber\\
&\textrm{subject to:}& 0\leq B_R\leq  B.
\end{eqnarray}
where
$
g_i(B_R, r)=N\beta^i\theta^i\frac{B-B_R}{\mu_1(B_R, r)M} - \Delta^i\frac{B-B_R}{\mu_1(B_R, r)M}\sum_{j=i+1}^TN\beta^j -N\beta^i\epsilon_1
$.
Note that $\hat{q}_i$ is a function of $B_R, r$. Again, the problem can be solved using two-dimensional exhaustive search to find the solution when we restrict the range of $r$ to be [\emph{\b{r}}, \emph{\={r}}]. We provide numerical results in Section \ref{sec:simulation}.

\section{Numerical Results}
\label{sec:simulation}

In this section, we use numerical results to evaluate the proposed hybrid pricing scheme. We first introduce the simulation setup and then discuss the results.

\subsection{Simulation Setup}

We assume there are a total of $N=100-1000$ SUs in the network. The default value of $N$ is 100. There are $M=100$ periods.
The expected available TVWS is $B=60$ MHz \cite{TVWS_availability1}\cite{TVWS_availability2}
We assume there are $T=5-30$ types, with 10 as the default value.
For $T=10$, we generate 5 different possible type distributions summarized in table \ref{tab:distributions}.
For all $T=5-30$, we also have a random type distribution where the number of SUs in each type is randomly generated.

\begin{table}[tp]
\renewcommand{\arraystretch}{0.9}
\caption{$N=100$ SUs in $T=10$ types.}
\label{tab:distributions}
\centering
\begin{tabular}{|c|c|c|c|c|c|c|c|c|c|c|}
  \hline
  $\theta^i$   & 1  & 2  & 3  & 4  & 5  & 6  & 7  & 8  & 9  & 10  \\
  \hline
  Distr. 1     & 10 & 10 & 10 & 10 & 10 & 10 & 10 & 10 & 10 & 10 \\
  \hline
  Distr. 2     & 1  & 3  & 5  & 7  & 9  & 11 & 13 & 15 & 17 & 19 \\
  \hline
  Distr. 3     & 19 & 17 & 15 & 13 & 11 & 9  & 7  & 5  & 3  & 1  \\
  \hline
  Distr. 4     & 2  & 6  & 10 & 14 & 18 & 18 & 14 & 10 & 6  & 2  \\
  \hline
  Distr. 5     & 18 & 14 & 10 & 6  & 2  & 2  & 6  & 10 & 14 & 18 \\
  \hline
\end{tabular}
\end{table}

When applying a two-dimensional search, we need to fix the range of $r$. Recall that $r$ is the money paid by the registered SUs to the DO. When $B_R=60, \mu_0=1$, the highest revenue a SU can obtain is $B_R/\mu_0\times\theta^{10}=600$. So $r$ should be within $[0, 600]$. We assume the smallest change in $r$ is $r_0=1$ unit. We assume the total TV bandwidth can only be allocated in the unit of $b_0=6$ MHz (the same as the spectrum span of a TV channel in the US). So the possible values of $B_R$ are $\{0, 6, 12, \cdots, 60\}$.

For the database maintenance cost we set $\alpha=1.2$ in $\phi_0(b)=\epsilon_0\cdot b^{\alpha}$ and we vary the parameter $\epsilon_0$ in the range $[0, 7]$. If $\epsilon_0$ is too high, DO will not have the incentive to consider the registration scheme. Since the marginal cost of one database access is very small compared with the bandwidth reservation cost, we set the default value of $\epsilon_1$ to be 0.

In the case of non-strategic SUs, we assume $\gamma^i$s are the same among all types. We use a single value $\gamma$ to denote the fraction of SUs preferring the registration scheme. We will set $\gamma$ to be either $0.2$ or $0.5$ in our simulations.

In this section, for convenience we refer to the Complete Information Scenario as \emph{CIS} and the Incomplete Information Scenario as \emph{IIS}.

\subsection{Simulation Results}

\subsubsection{Impact of SU strategy}

In this section, we compare the two cases: non-strategic SUs in CIS and strategic SUs in CIS to show the impact of SU strategy on the pricing solution.

First, in Fig. \ref{fig:epsilon0-B_R_non-strategic}, we fixed the SU type distribution to be Distr. 1. We set $\epsilon_1$ to be either 0.05 or 0.1. We then vary the reservation cost $\epsilon_0$ and plot the optimal $B_R$ obtained from three cases: non-strategic SUs with $\gamma=0.2$, non-strategic SUs with $\gamma=0.5$ and strategic SUs, all under CIS. From Fig. \ref{fig:epsilon0-B_R_non-strategic}, we can see that with non-strategic SUs, the DO tends to reserve less bandwidth for the reservation scheme under the same $\epsilon_0$ and $\epsilon_1$. That is because when the SUs are non-strategic players, they decide whether to pay the registration fee only based on the prior knowledge of the number of SUs in the registration scheme which is predetermined. Therefore, many SUs will overestimate the registration fee, so that the DO cannot increase the registration fee to obtain higher profit. As a result, the DO prefers the service plan scheme and reserves less bandwidth.
\begin{figure}[t]
\centering
\subfloat[$\epsilon_1=0.05$.]
{\centering \includegraphics[width=1.6in]{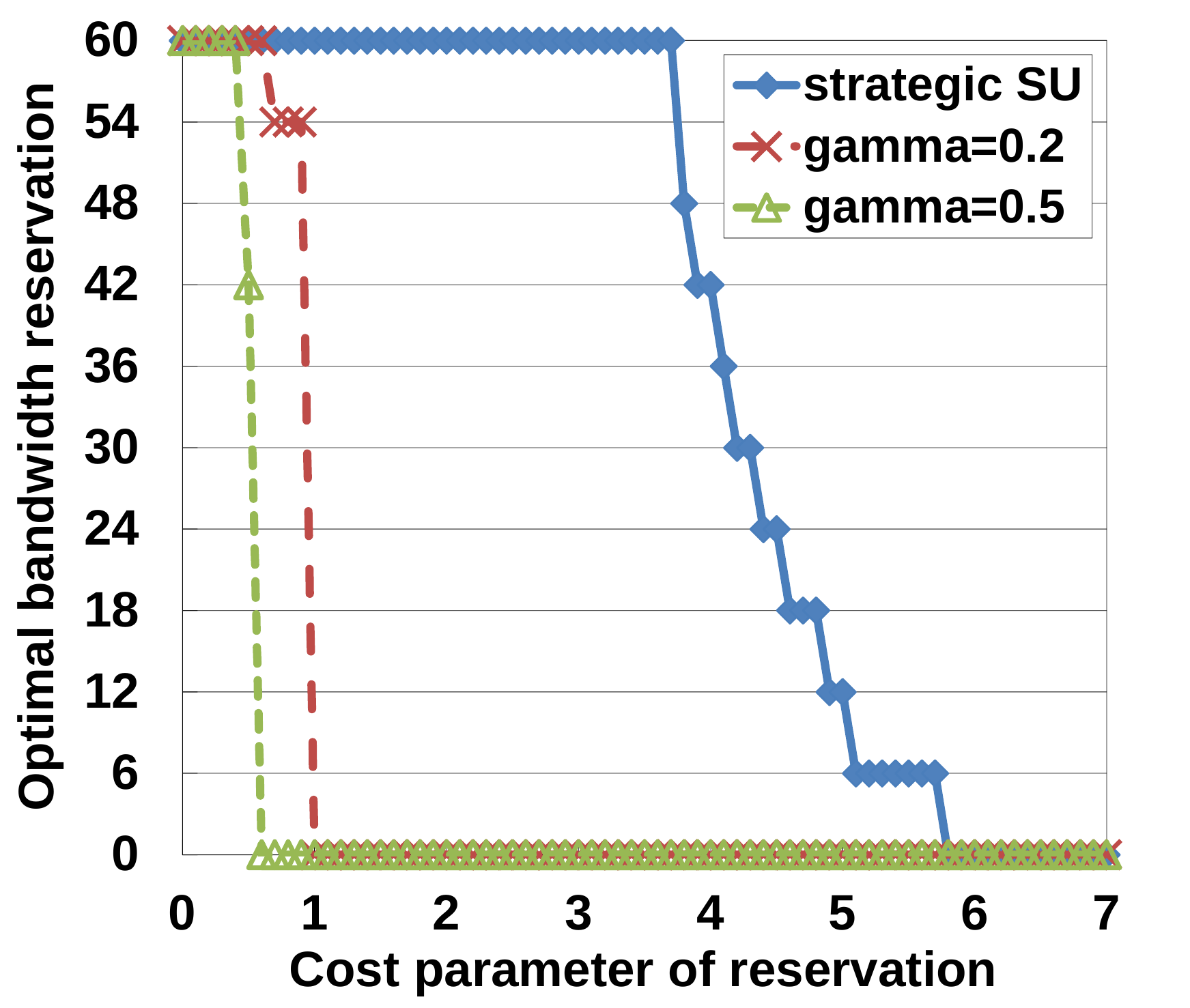}         
\label{fig:epsilon0-B_R_non-strategic_ep1=0.05}}%
\subfloat[$\epsilon_1=0.1$.]
{\centering\includegraphics[width=1.6in]{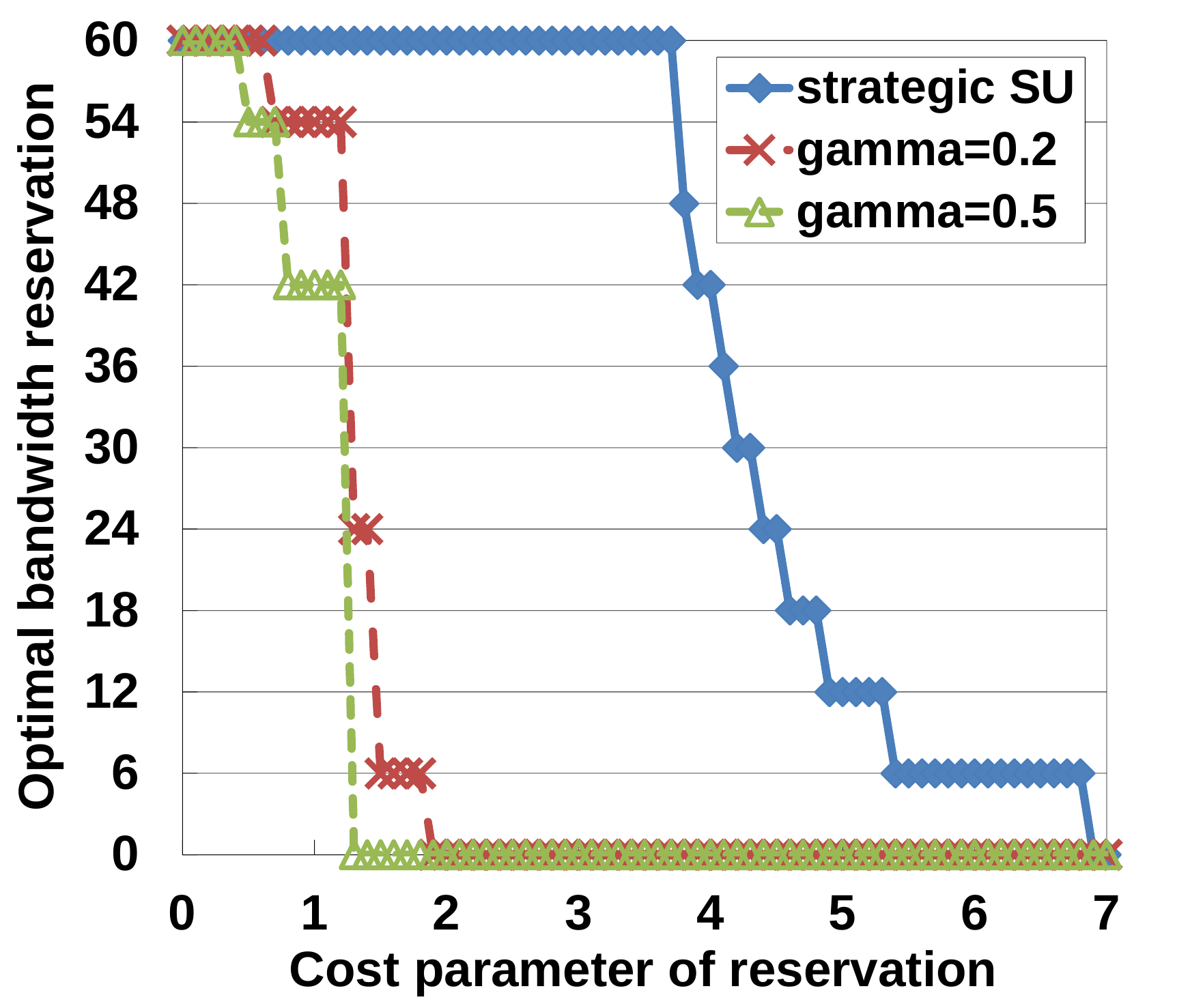}           
\label{fig:epsilon0-B_R_non-strategic_ep1=0.1}}%
\caption{Impact of SU strategy on bandwidth reservation.}
\label{fig:epsilon0-B_R_non-strategic}
\vspace{-0.0cm}
\end{figure}

Second, in Fig. \ref{fig:strategic_vs_non-strategic}, we fixed the SU type distribution to be Distr. 1 and set $\epsilon_1$ to be 0.05. We vary the reservation cost $\epsilon_0$ and plot the utility of DO and the average utility of SUs given the optimal pricing parameters. From Fig. \ref{fig:strategic_vs_non-strategic_UDO}, we can see that in both scenarios, the DO utility first decreases and then remains the same after $B_R=0$. When $\epsilon_0$ is higher, the DO utility in the non-strategic scenario is higher compared with that in the strategic SU scenario. That's because when $\epsilon_0$ is higher, in the strategic scenario, the profit DO makes from the registered SUs become smaller however, most of DO profit in the non-strategic SU scenario comes from the non-registered SUs.
We can also see from Fig. \ref{fig:strategic_vs_non-strategic_USU} that when SUs are strategically player, their average utility are still non-zero when $\epsilon_0>0.5$. That is because, in the non-strategic SU scenario, when $\epsilon_0>0.5$, $B_R=0$, all SUs can only get zero utility. However, in the strategic SU scenario, some SUs can choose the registered scheme to obtain non-zero utility.
From the comparison, we can conclude that with more intelligent SU strategy, SUs achieve higher average utility.
\begin{figure}[t]
\centering
\subfloat[DO utility.]{
\centering \includegraphics[width=1.6in]{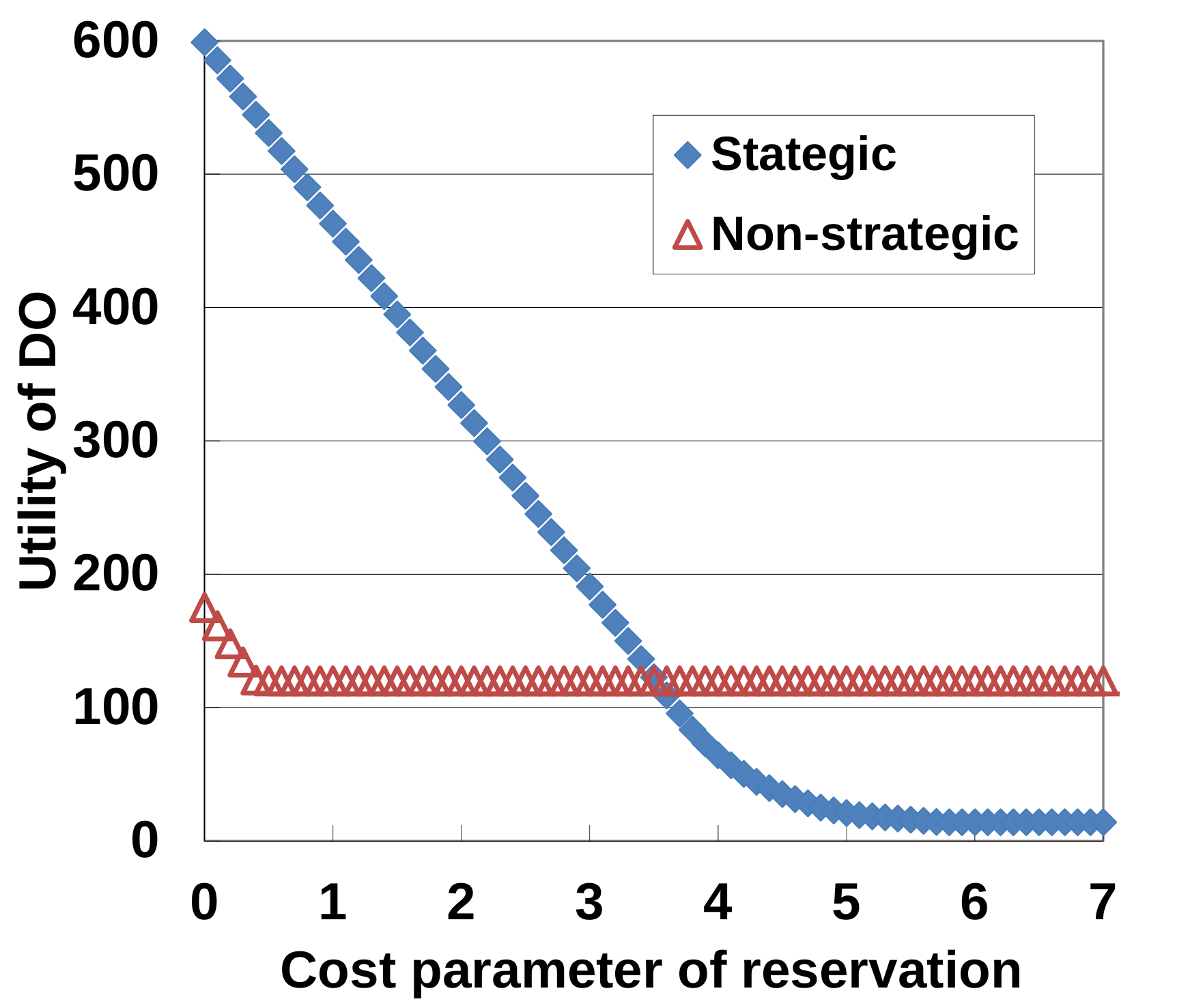}           
\label{fig:strategic_vs_non-strategic_UDO}}%
\subfloat[Average SU utility.]{
\centering\includegraphics[width=1.6in]{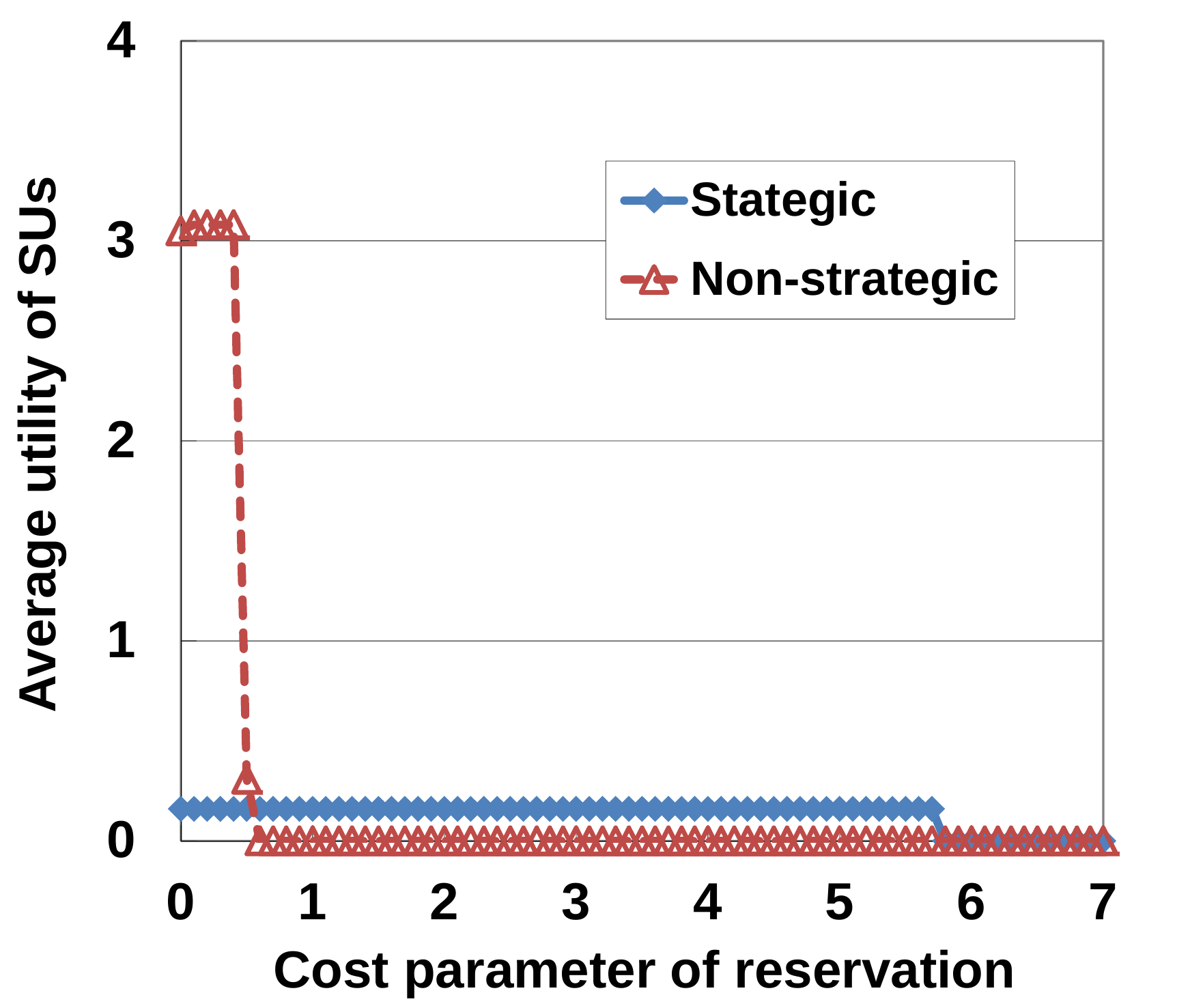}            
\label{fig:strategic_vs_non-strategic_USU}}%
\caption{Impact of SU strategy on utility.}
\label{fig:strategic_vs_non-strategic}
\vspace{-0.0cm}
\end{figure}

\subsubsection{Benefit of Hybrid Pricing Scheme}

We compared the DO's revenue when applying hybrid and uniform pricing schemes. There are two uniform pricing schemes: registration only or service plan only. We set $epsilon_0=3.0$. In the case of registration scheme only, $B_R=60$ while in the case of service plan scheme only, $B_R=0$. We show the maximum possible revenues the DO can obtain in type distributions 1-3 and random distribution in Fig. \ref{fig:hybrid}. The results from random distribution is an average of 100 runs. We can see that the hybrid pricing scheme provides higher revenue for the DO. By offering hybrid pricing schemes, the DO has a new degree of freedom to tune the bandwidth segmentation to increase its profit. We will evaluate the impact of $epsilon_0$ in the next evaluation.

\subsubsection{Impact of the bandwidth reservation cost}

To show the impact of the bandwidth reservation cost on the optimal bandwidth reservation
$B_R$, we vary $\epsilon_0$ from 0 to 5.2 and plot the optimal $B_R$ under each $\epsilon_0$ in both CIS and IIS in Fig. \ref{fig:epsilon0-B_R_scenario0} and Fig. \ref{fig:epsilon0-B_R_scenario1}, respectively. In both information scenarios, $B_R$ decreases with the increase of $\epsilon_0$. That is because with greater $\epsilon_0$, the reservation cost for the same bandwidth is higher, so the DO prefers to allocate more bandwidth for unregistered users. We can see that when $\epsilon_0\geq 5.2$, $B_R=0$ for all cases, which means no bandwidth is reserved for the registration scheme.

We can also observe from Fig. \ref{fig:epsilon0-B_R} that in Distr. 2 $B_R$ decreases the fastest among the three distributions. That is because in Distr. 2, there are more SUs of higher types, who are more likely to have $g_i>0$ in Eqn. (\ref{eqn:max_udo3}), thus DO can benefit more from the service plan scheme.

It is also worth noting that in the CIS, $B_R$ decreases to zero sooner compared with that in the IIS. That is because in CIS, DO can design the query plans for individual unregistered SUs to make more profit.

\begin{figure*}[tp]
\centering
\begin{minipage}{2.0in}
\centering
\includegraphics[width=1.9in]{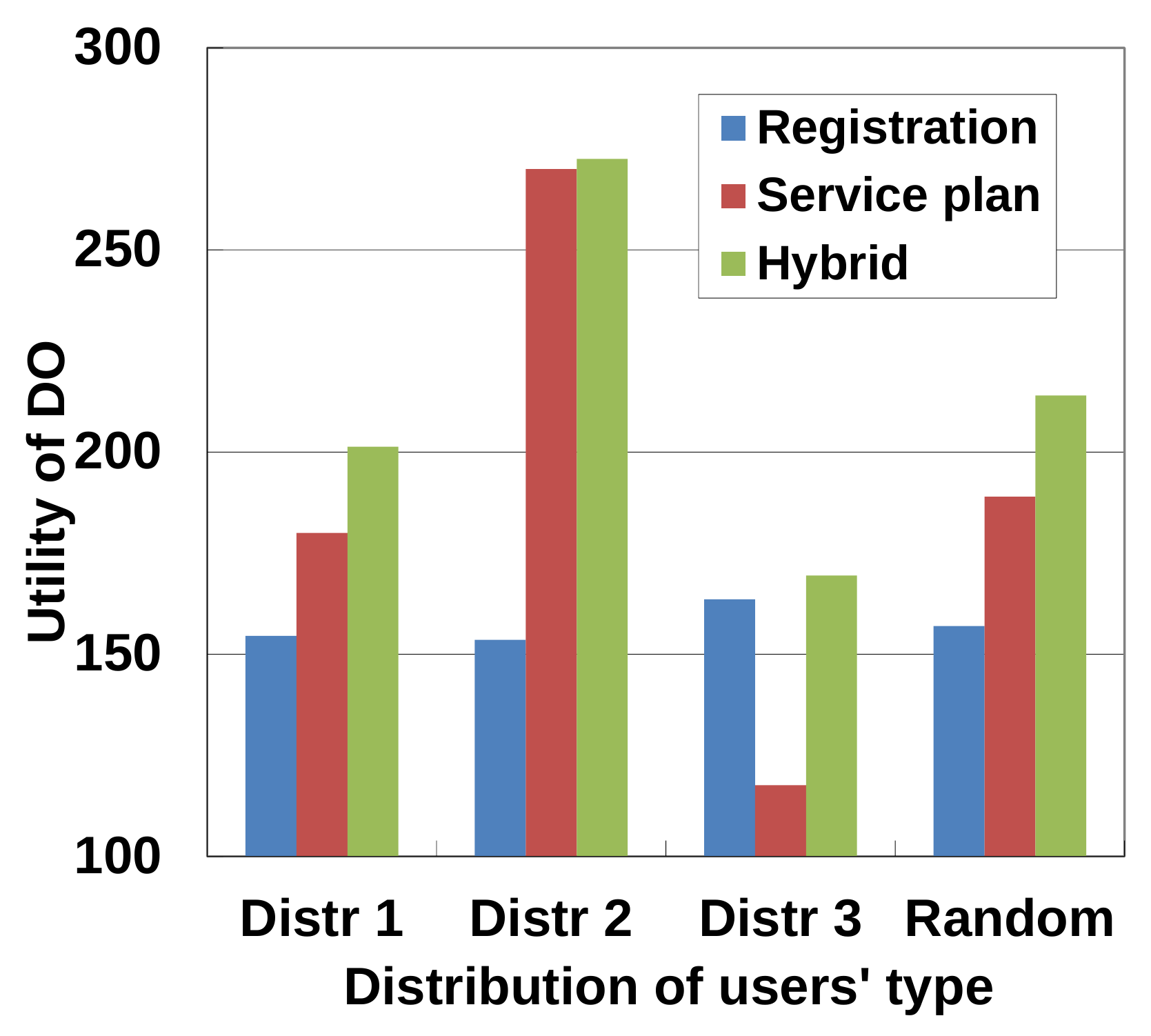}
\caption{Comparison between hybrid and uniform pricing.}
\label{fig:hybrid}
\end{minipage}
\hspace{1.2cm}
\begin{minipage}{4.0in}
\centering
\subfloat[Complete information.]{\centering
\centering \includegraphics[width=1.9in]{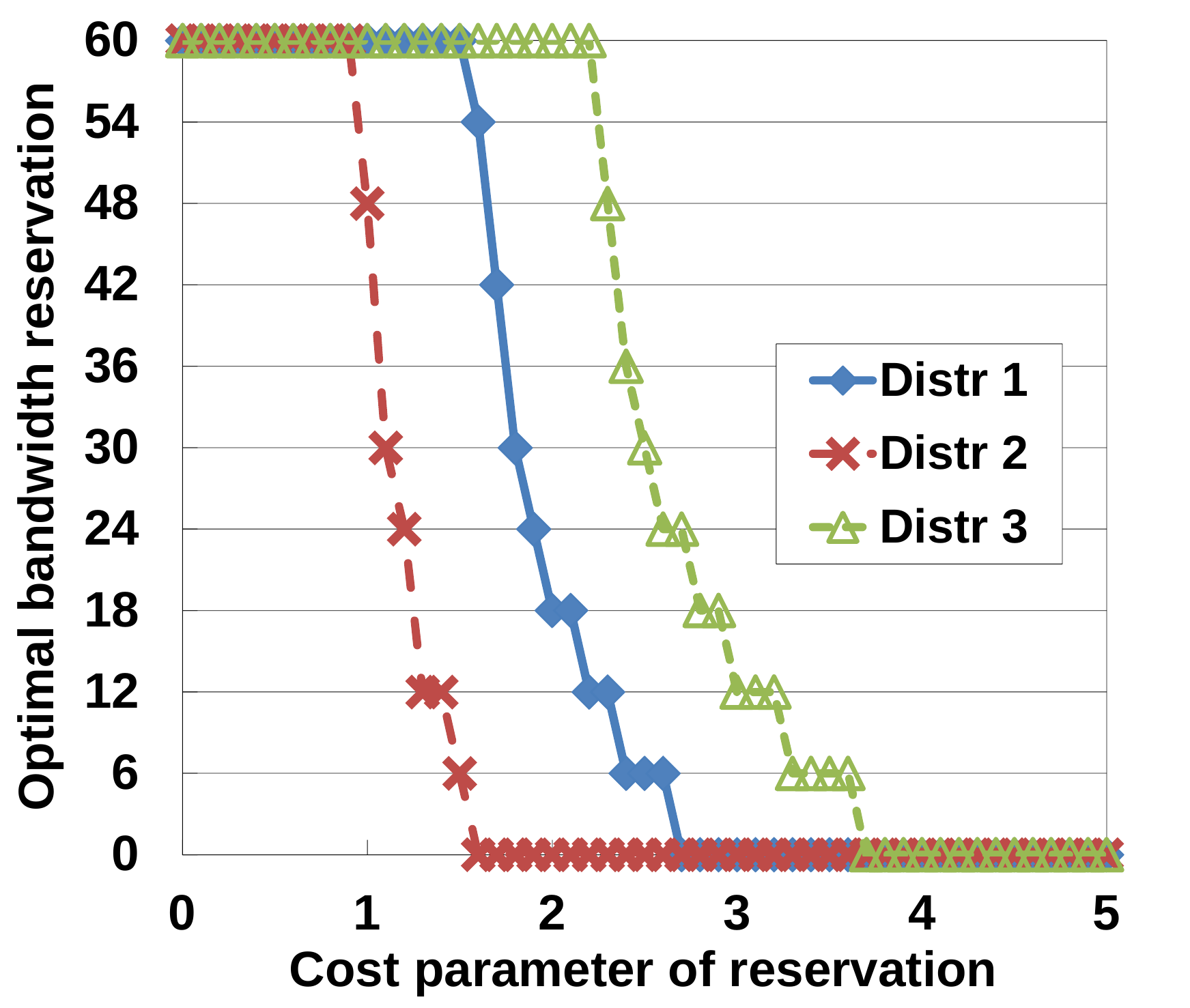}
\label{fig:epsilon0-B_R_scenario0}}%
\hspace{0.0cm}
\subfloat[Incomplete information.]{\centering
\centering\includegraphics[width=1.9in]{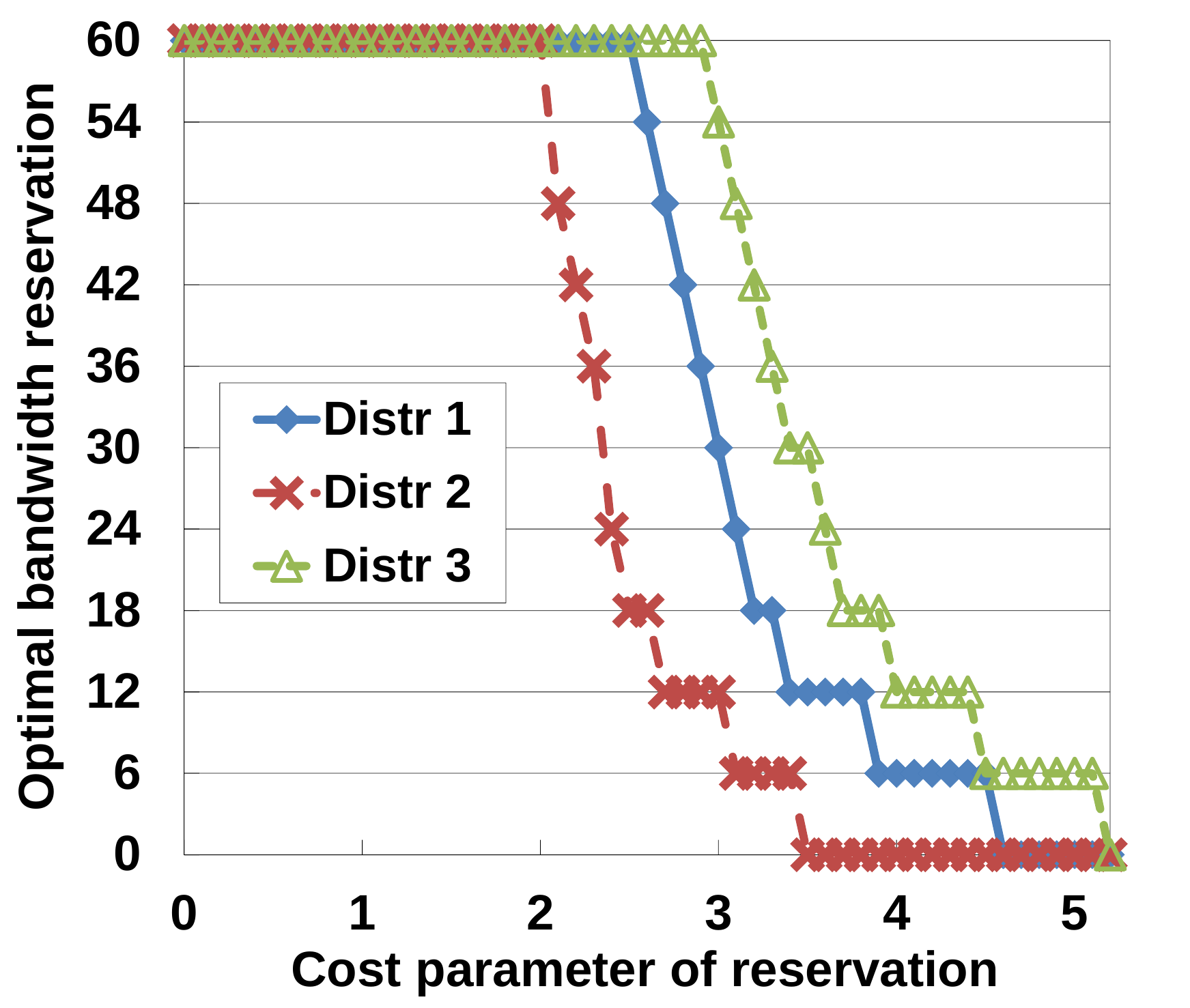}
\label{fig:epsilon0-B_R_scenario1}}%
\caption{Impact of bandwidth reservation cost.}
\label{fig:epsilon0-B_R}
\vspace{-0.0cm}
\end{minipage}
\end{figure*}

\subsubsection{Impact of the DO's knowledge of SUs' personal information}

In Fig. \ref{fig:information}, we show the impact of the DO's knowledge of the SUs' personal information on the utility of the DO and the SUs. The SUs' types are set to distribution 1.

We can see from Fig. \ref{fig:information_UDO} that the DO has higher utility in the CIS than that in the IIS, when $\epsilon_0\geq 1.7$. That is because when $\epsilon_0\geq 1.7$, $B_R<60$ in CIS. As a result, the bandwidth for the service plan scheme is non-zero and the DO can make more revenue from unregistered SUs in the CIS.

We can see from Fig. \ref{fig:information_USU} that the SUs have higher average utility in the IIS than that in CIS, when $\epsilon_0>2.7$. That' because when $\epsilon_0>2.7$, $B_R<60$ in IIS. As a result, the bandwidth for the service plan scheme is non-zero and the SUs can get non-negative utility when choosing the service plan scheme.

We can also notice that when $\epsilon_0<1.7$ in Fig. \ref{fig:information_UDO} and when $\epsilon_0<2.7$ in Fig. \ref{fig:information_USU}, $B_R=60$, which means all bandwidth is reserved for registered SUs. As the DO design registration scheme to admit only registered SUs with the highest type, the utilities obtain in both information scenario are the same.

In summary, with more knowledge of the SUs' information, the DO enjoys higher utility. On the other hand, the hidden information provides the SUs higher utility.

\begin{figure}[t]
\centering
\subfloat[DO utility.]{\centering
\centering \includegraphics[width=1.6in]{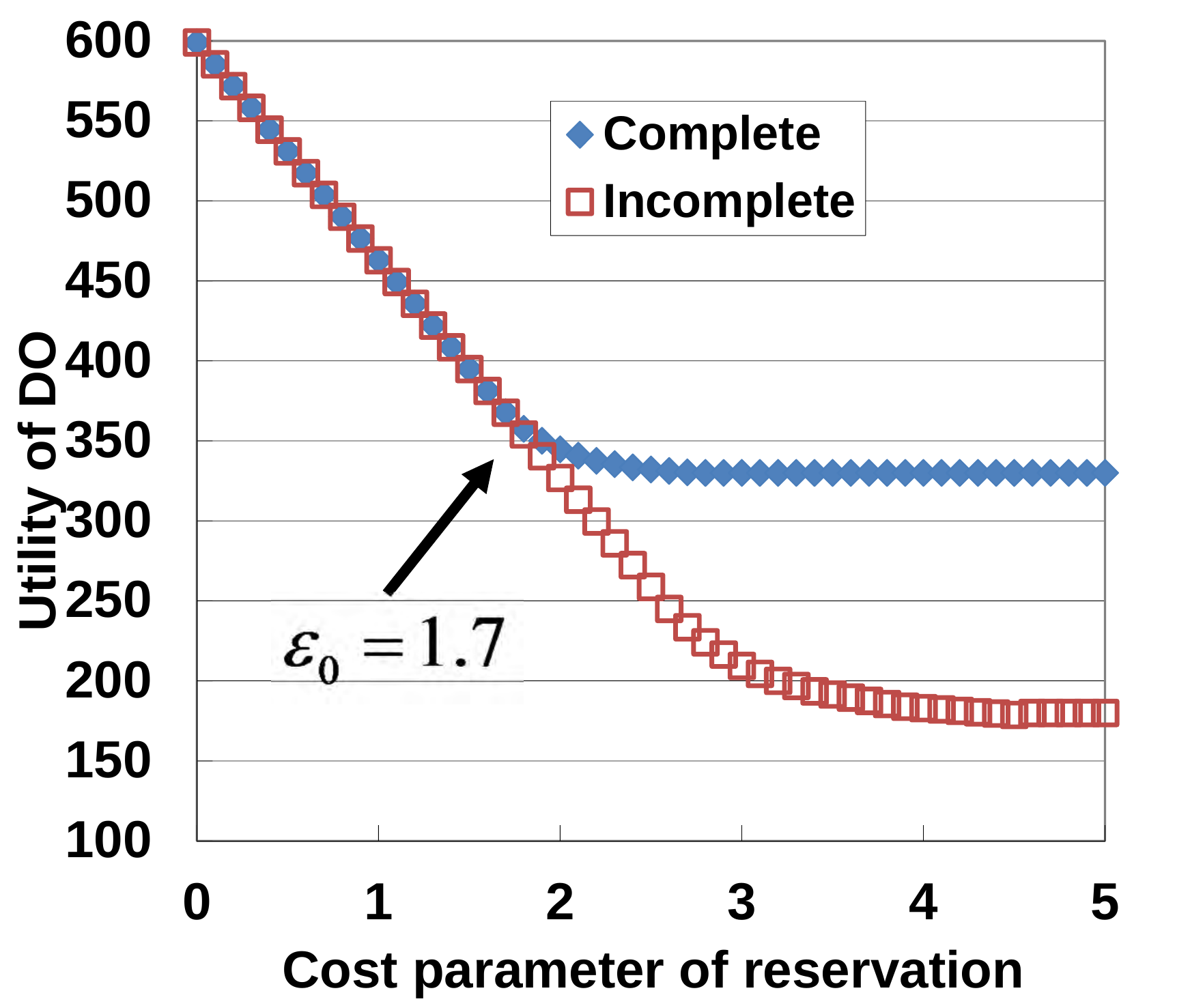}          
\label{fig:information_UDO}}%
\subfloat[Average SU utility.]{\centering
\centering\includegraphics[width=1.6in]{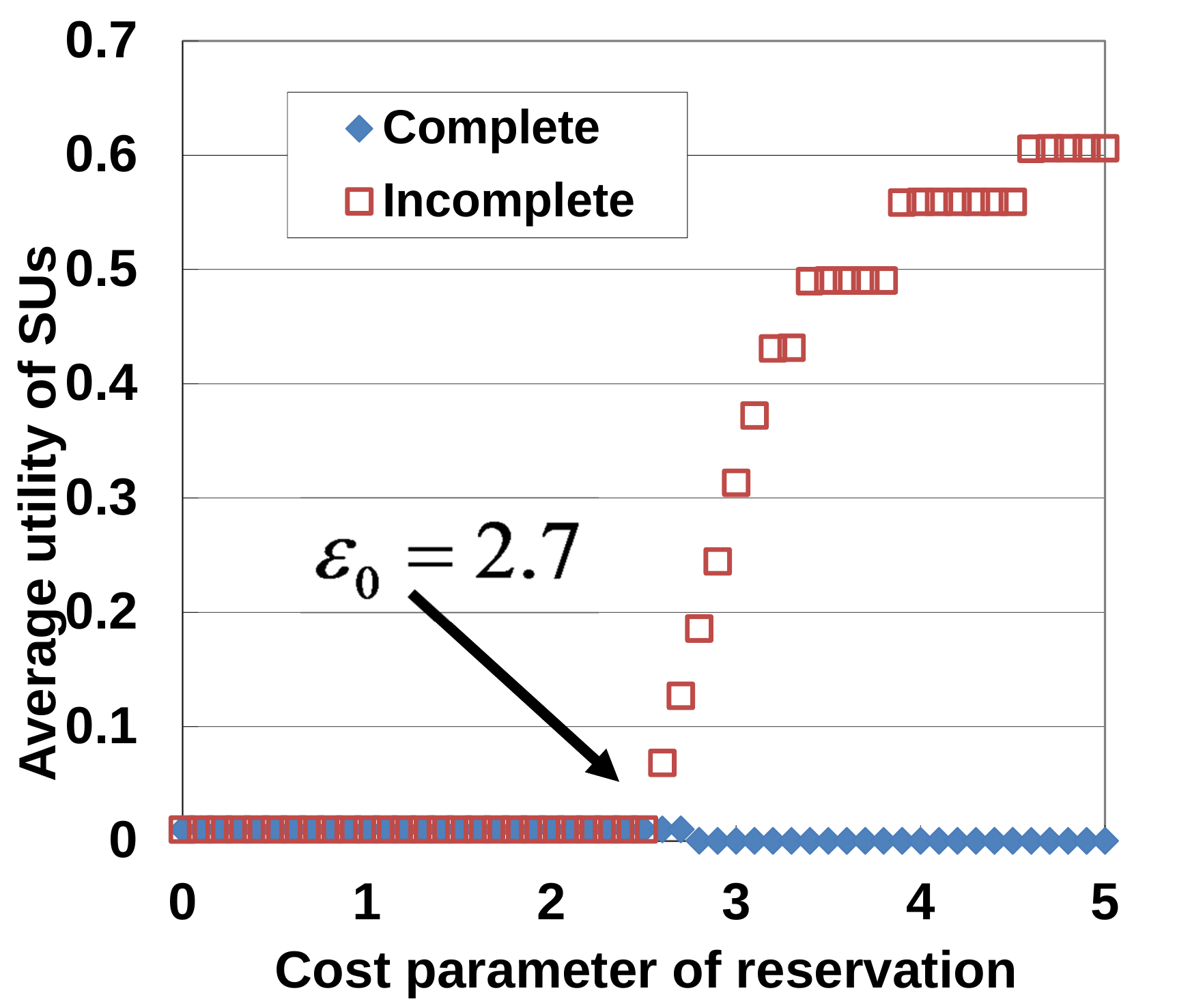}           
\label{fig:information_USU}}%
\caption{Impact of information.}
\label{fig:information}
\vspace{-0.0cm}
\end{figure}

\subsubsection{Impact of the SUs' type distribution}

To show the impact of the SUs' type distribution on the contract design, we fixed $B_R=30$, $r=200$, $\epsilon_0=3.0$ and compute the $q_i$ and $p_i$ in the service plans for 50 unregistered SUs in the 5 different SU type distributions. Note that the parameters may not be the optimal pricing parameters for each distribution, the use of the same parameters for all the distributions allows us to see the impact of SU type distribution separately. Fig. \ref{fig:contract} shows the results for $q_i$ and $p_i$. In Fig. \ref{fig:contract_query}, we can see that when the number of SUs of lower types is higher, (compared distr. 3 and distr. 1), DO tends to assign non-zero queries for lower type SUs. That is because a query $q_i$ is determined by a factor of $\sum_{j=i+1}^TN\beta^j$ in Eqn. (\ref{eqn:max_udo3}). In Fig. \ref{fig:contract_price} we can see that if non-zero $q_i$ is assigned to a lower type, a lower price $p_i$ should be charged. This is because, the price should follow the individual rationality constraint (Eqn. (\ref{eqn:IR})) for such SUs.

\begin{figure}[t]
\centering
\subfloat[Query in the service plans.]{\centering
\centering \includegraphics[width=1.6in]{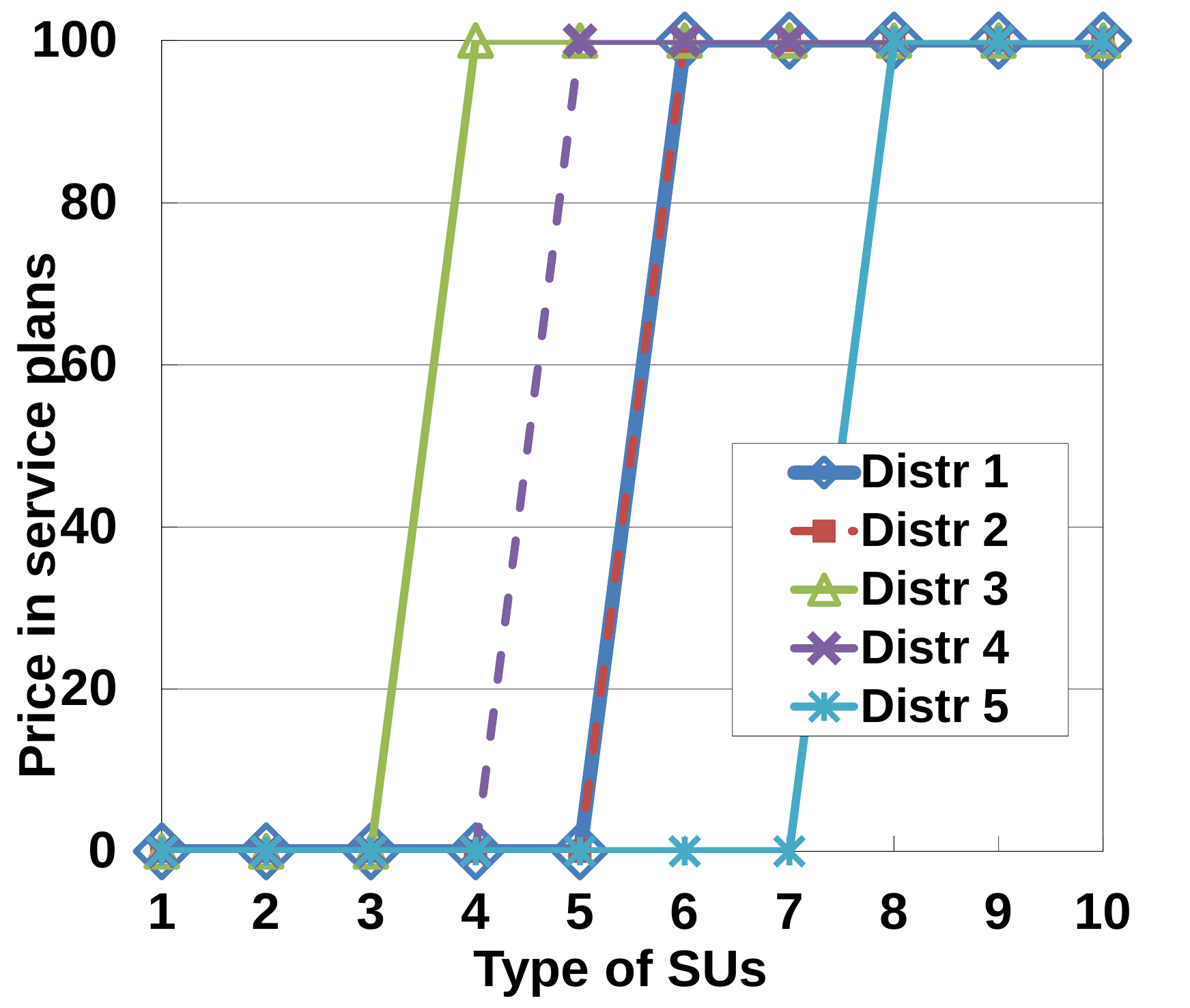}           
\label{fig:contract_query}}%
\subfloat[Price in the service plans.]{\centering
\centering\includegraphics[width=1.6in]{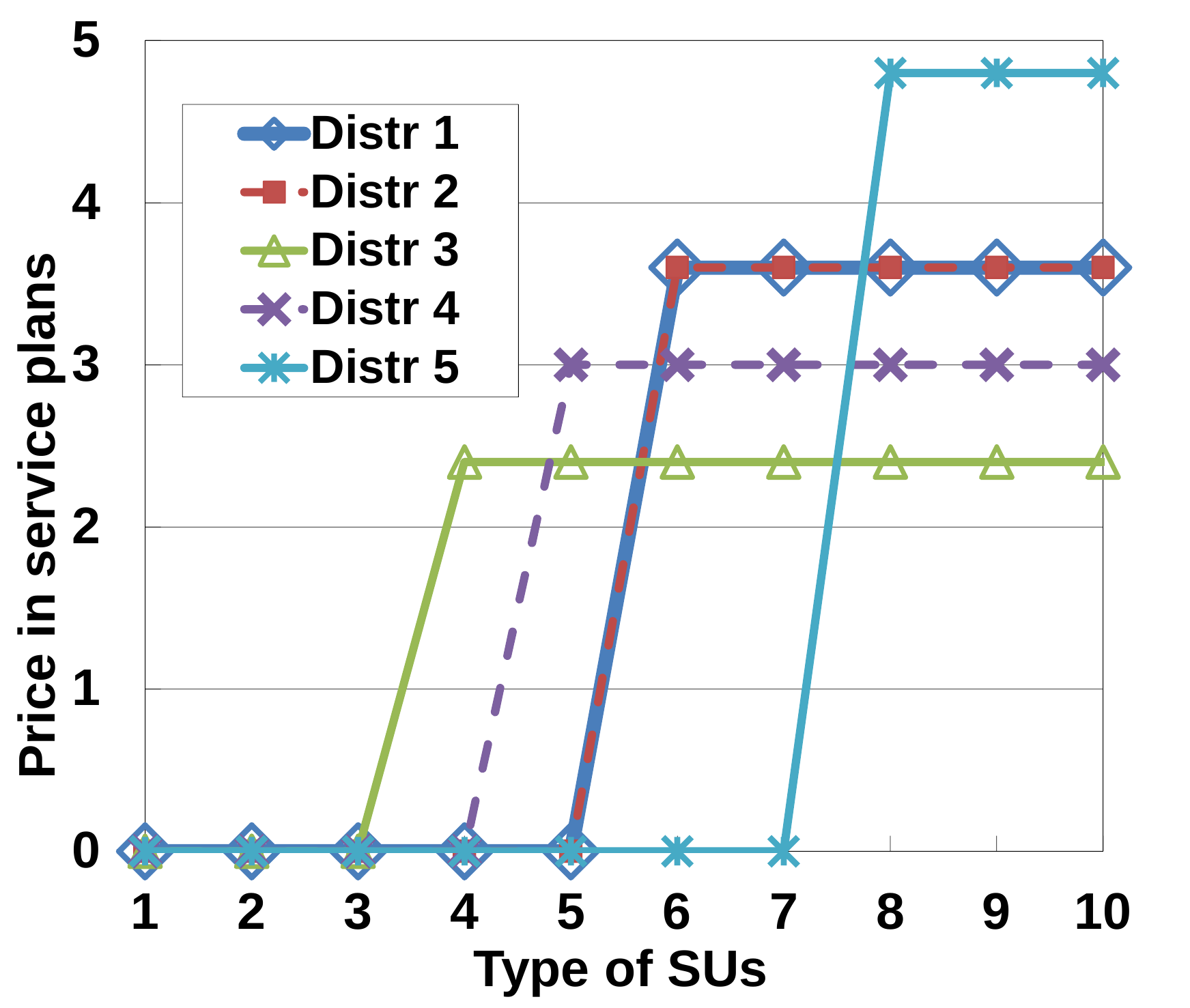}            
\label{fig:contract_price}}%
\caption{Query plans under different SU type distributions.}
\label{fig:contract}
\vspace{-0.0cm}
\end{figure}

\subsubsection{Impact of the suboptimal query assignment}

We implemented the brunching and ironing algorithm \cite{contract_book} for the contract design. We assume at stage I, the DO and the SUs still use Algorithm \ref{alg:valid_q} to estimate the contracts to ensure the existence of an NE. We generate random SU type distributions with $N=100$, $T=5-30$ and set $B_R=0$, which indicates that all the $N$ SUs choose the service plan scheme. We repeated the test for 100 times. We show the average utility of the DO obtained from query assignment algorithms in Fig. \ref{fig:suboptimal}. For a different $T$, Algorithm \ref{alg:valid_q} achieves the DO utility within 90\% of that obtained from the optimal algorithm.

\subsubsection{Convergence time to an NE under incomplete information scenario}

To check the convergence time to NE via Algorithm \ref{alg:registration_incomplete_info}, we fix $B_R=30$, $r=200$ and generate $N=100-1000$ SUs with uniformly distributed types. Under each $N$, we repeat Algorithm \ref{alg:registration_incomplete_info} 200 times. Fig. \ref{fig:convergence} shows the average steps of improvement needed to achieve an NE. It is clear that the number of steps is bounded by $N\cdot(N+1)$, which verifies Theorem \ref{thm:NDRG_property2}. We may prove a tighter bound for the convergence time in our future work.

\begin{figure}[tp]
\centering
\begin{minipage}{1.6in}                                     
\centering
\includegraphics[width=1.6in]{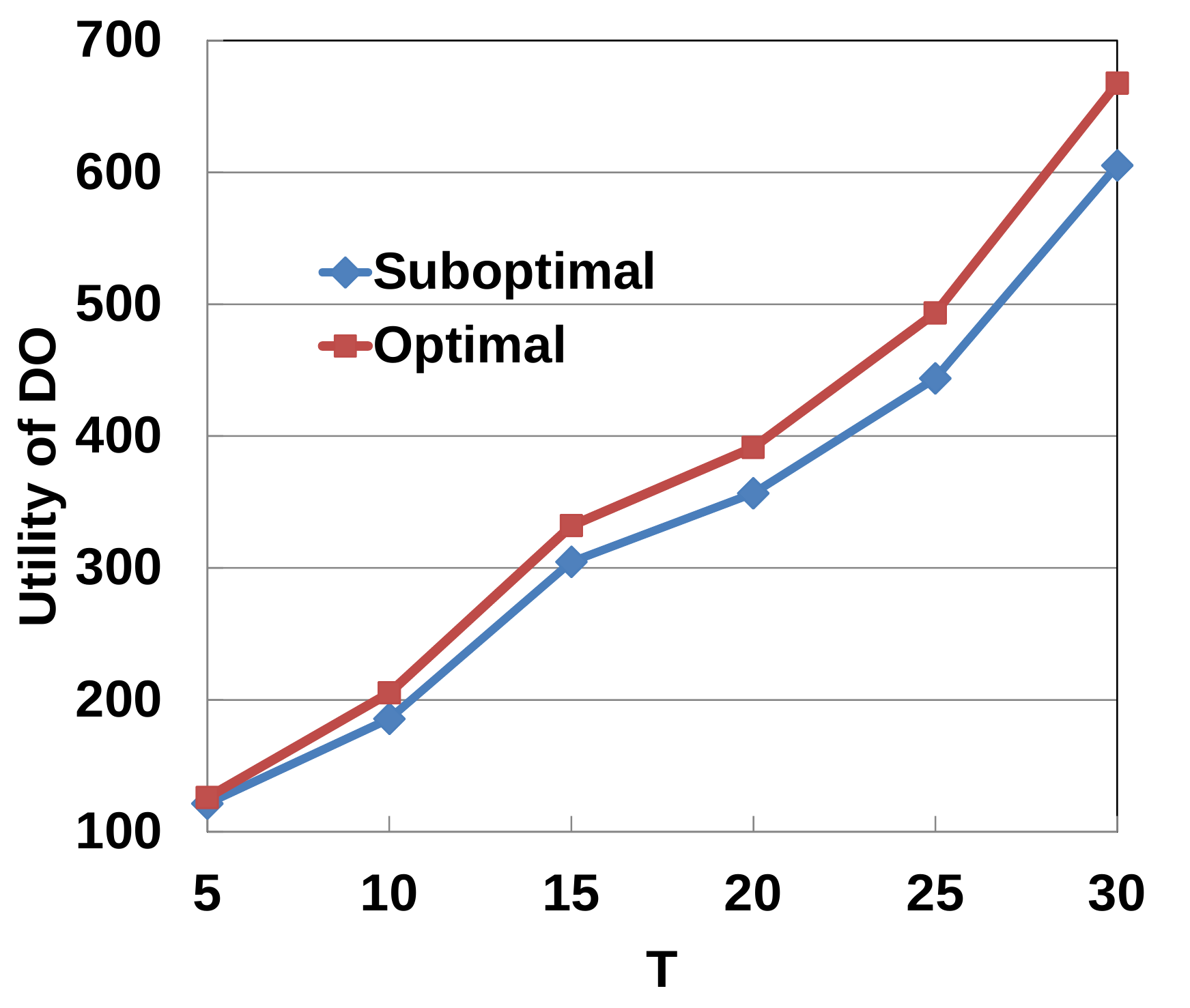}          
\caption{Impact of suboptimal query assignment.}
\label{fig:suboptimal}
\end{minipage}
\begin{minipage}{1.6in}                                     
\centering
\includegraphics[width=1.6in]{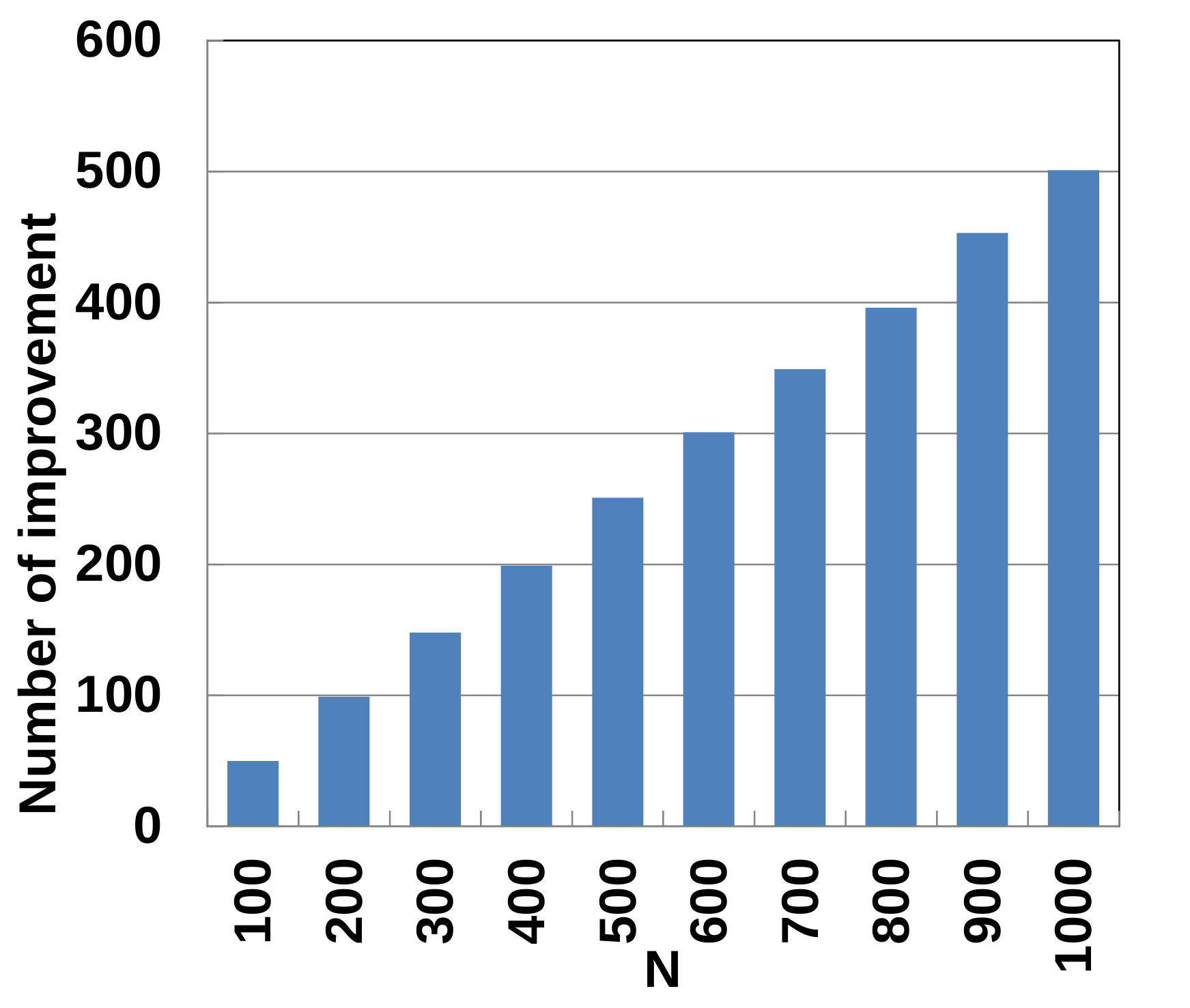}         
\caption{Convergence time to an NE under IIS.}
\label{fig:convergence}
\end{minipage}
\vspace{-0.0cm}
\end{figure}

\section{Discussion and Future Works}

In this paper, to make the problem solvable, we have made several simplifications in the system model. There are several interesting ways to further extend this work.

First, the registration scheme considered in this paper admits a uniform registration fee. In a more flexible setting, the registered SUs can pay different amounts of registration fees and enjoy different shares of the reserved bandwidth. Similarly with the query plan scheme, the DO can offer several levels of registration fees.

Second, the pricing framework considered in this paper is not dynamic. The pricing parameters are determined at the beginning once and will not change. In a more dynamic setting, the DO can redesign new pricing parameters for the following time durations and the SUs can also re-select the pricing scheme given the observation in previous time durations and the newly announced pricing parameter. In this case, the interactions between the DO and the SUs become a repeated game. Also, the SUs may dynamically join and leave the query plan. Therefore, the SUs should consider the expectation of future SU behaviors when choosing their pricing schemes.

\section{Related Works}
\label{sec:related_works}

Most existing works on geo-location databases can be classified into two categories. Some works focus on the design of geo-location database to protect primary users. In \cite{database_dyspan08}, Gurney et al. discussed the methods to calculate the protection area for TV stations. In \cite{database_dyspan11}, Murty et al. designed a database-driven white space network based on measurement studies and terrain data. Some other works focused on the networking issue with the assumption that the database is already set up. In \cite{database}, Feng et al. presented a white space system utilizing a database. In \cite{database_icdcs12}, Chen et al. considered the channel selection and access point association problem.
One recent work \cite{database_icc12} also address the business model related to the geo-location database. In \cite{database_icc12}, the authors proposed that the geo-location database acts as a spectrum broker reserving the spectrum from spectrum licensees. They considered only one pricing scheme which is similar to the registration scheme discussed in our paper.
Compared to our previous work \cite{databasePricing}, in this paper, we further extend the scenario to non-strategic SUs and compared the pricing schemes with non-strategic and strategic SUs under complete information scenario. We also extend our theoretical analysis and numerical evaluations.

Many works also focus on the economic issue of dynamic spectrum sharing. In \cite{add1}, the pricing-based spectrum access control is investigated under secondary users competitions. In \cite{add2}, spectrum pricing with spatial reuse is considered. Contract theory is utilized in the scenarios where the spectrum buyers have hidden information. In \cite{JSAC10}, Gao et al. leveraged contract theory to analyze the spectrum trading between primary operator and SUs. In \cite{dyspan11Duan}, contract theory is applied to the cooperative communication scenario. In this paper, we also model the service plan design with contract theory. However, due to the co-existence of hybrid pricing schemes, there is uncertainty about the number of SUs choosing the contract items, which is different from existing works.

There are some works focus on the hybrid pricing of other limited resources. In \cite{icdcs12segmentation}. Wang et.al study the problem of capacity segmentation for two different pricing schemes for cloud service providers. One key difference between our work and \cite{icdcs12segmentation} is that the strategic SUs considered in our paper can dynamically choose between pricing schemes. While in \cite{icdcs12segmentation}, the users are pre-categorized into different pricing scheme before designing the pricing schemes.

\section{Conclusions}
\label{sec:conclusion}

In this paper, we consider a hybrid pricing model for TVWS database. The SUs can choose between the registration and the service plan scheme. We investigate scenarios where the SUs can be either non-strategic or strategic players and the DO has either complete and incomplete information of the SUs. In the non-strategic SU scenario, we model the competitions among the SUs as non-cooperative game and prove the existence of an NE in two different scenarios by showing that the game is an unweighted congestion game. We model the pricing for unregistered SUs with contract theory and derive suboptimal query plans for different types of SUs. Based on the SUs' pricing scheme choices, the DO optimally determines the bandwidth segmentation and pricing parameters to maximize its profit. We have conducted extensive simulations to obtain numerical results and verify our theoretical claims.

\section*{Acknowledgement}
The research was support in part by grants from 973 project 2013CB329006, China NSFC under Grant 61173156, RGC under the contracts CERG 622410, 622613, HKUST6/CRF/12R, and M-HKUST609/13, as well as the grant from Huawei-HKUST joint lab.

\end{document}